\documentclass[%
 amsmath,amssymb,
 aps, physrev,
]{revtex4-2}

\usepackage{amsfonts}
\usepackage{dsfont}
\usepackage{amsmath}
\usepackage{amssymb}
\usepackage{bm}
\usepackage{subcaption}
\usepackage{comment}

\usepackage{enumitem}

\setenumerate{label=\roman*), ref=(\roman*)}

\usepackage{xcolor}
\definecolor{myblue}{HTML}{006AA3}
\definecolor{mygreen}{HTML}{5B9A64}
\definecolor{myred}{HTML}{C1121F}
\definecolor{myorange}{HTML}{FAA00F}
\definecolor{mypurple}{HTML}{6D4C94}

\usepackage{tikz}
\usepackage{pgfplots}
\pgfplotsset{compat=1.18}
\usepgflibrary{arrows.meta} 

\usetikzlibrary{arrows.meta,positioning,backgrounds,shapes,calc,fit,hobby,decorations.markings,3d,calc}
\pgfdeclarelayer{background}
\pgfsetlayers{background,main}
\usepgfplotslibrary{fillbetween}
\tikzset{
    node style/.style={circle, draw, fill=white, inner sep=1pt, minimum size=1.5em},
    main line/.style={thick}
}

\usepackage{amsthm}

\newtheorem{theorem}{Theorem}
\newtheorem{definition}{Definition}
\newtheorem{proposition}{Proposition}

\newtheorem{standassumption}{Standing Assumption}
\newtheorem{remark}{Remark}

\usepackage[hidelinks]{hyperref}
\usepackage[ruled,vlined]{algorithm2e}

\newcommand{\Ecal}{\mathcal{E}}

\newcommand{\Gcal}{\mathcal{G}}

\newcommand{\Ncal}{\mathcal{N}}

\newcommand{\Vcal}{\mathcal{V}}

\newcommand{\norm}[1]{\left\| #1 \right\|}

\newcommand{\R}{\mathbb{R}}

\newcommand{\one}{\mathbf{1}}

\newcommand{\zero}{\mathbf{0}}
\usepackage{booktabs}
\usepackage{multirow}

\begin{document}


\title{Social learning community detection with nonlinear interaction}


\author{Anthony Couthures}
\affiliation{UCLouvain, ICTEAM Institute, Louvain-la-Neuve, 1348, Belgium.}

\author{Athira Varma Jayakumar}
\affiliation{DISTek Integration, Inc., 6612 Chancellor Dr., Ste. 600, Cedar Falls, Iowa 50613, USA}

\author{Vineeth Satheeskumar Varma}
\affiliation{Universit\'e de Lorraine, CNRS, CRAN, F-54000 Nancy, France.}

\author{Irinel-Constantin Mor\u{a}rescu}
\affiliation{Universit\'e de Lorraine, CNRS, CRAN, F-54000 Nancy, France.}

\author{Samson Lasaulce}
\affiliation{Universit\'e de Lorraine, CNRS, CRAN, F-54000 Nancy, France.}

\author{Antoine Girard}
\affiliation{Universit\'e Paris-Saclay, CNRS, CentraleSup\'elec, Laboratoire
    des signaux et syst\`emes, F-91190, Gif-sur-Yvette, France.}


\thanks{This work has been submitted to the IEEE for possible publication. Copyright may be transferred without notice, after which this version may no longer be accessible.}


\date{\today}

\begin{abstract}
    Conventional community detection requires centralized network data, making it unsuitable for distributed or privacy-preserving systems. In this paper, we demonstrate that macroscopic graph partitioning can emerge purely from strictly local, privacy preserving interactions driven by social learning. By reframing clustering as a symmetry-breaking process within nonlinear opinion dynamics, we show that exchanging saturated state dependent signal (like public actions) forces a network to naturally fracture along its sparsest cuts. We mathematically establish the spectral conditions under which dense core communities lock into stable, polarized states, robustly resisting external influence. To apply this mechanism, we propose three decentralized algorithms, leading up to the Score-based Edge Reliability (SER) framework. By evaluating network ties across multiple independent discussion topics, SER statistically bypasses the errors of traditional greedy bisections and naturally isolates structurally ambiguous frontier nodes. Validations on the ABCD benchmark and the real-world Ngogo chimpanzee network confirm that our fully decentralized approach matches the accuracy of globally optimized heuristics (e.g., Louvain, Leiden) up to a theoretical limit of detectable graphs.
\end{abstract}

\maketitle

\section{Introduction}
\label{sec:introduction}

The organization of complex networks is frequently characterized by an intermediate-scale structure consisting of communities: groups of nodes that are densely connected internally while being sparsely connected to the rest of the network. Identifying these structures is fundamental to understanding the functional properties, resilience, and information flow within a system \cite{fortunatoCommunityDetectionGraphs2010,liComprehensiveReviewCommunity2024}. While exact modularity optimization is computationally intractable, heuristic methods provide excellent approximations for community detection relying on three main paradigms: spectral clustering (e.g., \cite{pothenPartitioningSparseMatrices1990,spielmanSpectralPartitioningWorks2007}), modularity maximization (e.g., \cite{blondelFastUnfoldingCommunities2008,seifikarCBlondelEfficientLouvainBased2020}), and probabilistic methods (e.g., \cite{dempsterMaximumLikelihoodIncomplete1977,mosselStochasticBlockModels2012}). However, the vast majority of standard clustering algorithms rely on a centralized computing architecture. Centralized approaches inherently require a global view of the network topology, meaning the complete adjacency matrix must be stored and processed by a single entity. In many modern distributed systems---such as autonomous robot swarms, decentralized social platforms, and privacy-preserving sensor networks---acquiring this global view is practically impossible because the network may be too large, or because communication costs and legal privacy constraints prohibit such data aggregation. This creates a need for decentralized algorithms where the global community structure emerges purely from local peer-to-peer interactions.

We can draw direct inspiration for such algorithms from the way communities organically form in human social systems. In real-world societies, individuals do not possess a global map of the social network to determine their group affiliations, nor do they have direct access to the private beliefs of their peers. Instead, community identification emerges as an individual learning process driven by observing the public actions or decisions of neighbors. Through interaction, individuals exchange opinions that can be interpreted as signals across various topics. If two people consistently influence one another and align their behavior---by sharing similar actions or communication patterns---across multiple independent subjects (even if they initially held opposing views), they exhibit strong structural affinity and effectively belong to the same social group. Conversely, individuals gradually discount or sever ties with peers whom they consistently fail to persuade or with whom their observable decisions persistently clash. Evaluating these local interactions across a spectrum of topics reveals the underlying community boundaries. Forming local consensus while pruning discordant ties drives this social learning process.

A natural mathematical framework for translating this decentralized sociological behavior into computation on graphs is the study of multi-agent dynamical systems, particularly consensus protocols and social learning models (as introduced in \cite{frenchFormalTheorySocial1956}), which connect naturally to opinion dynamics models. In a standard linear DeGroot model \cite{degrootReachingConsensus1974}, agents repeatedly average their states with those of their neighbors. While effective for distributed agreement, linear consensus inevitably converges to a homogeneous global state for any connected graph \cite{bergerNecessarySufficientCondition1981}, which eliminates the structural diversity necessary to partition the network. To sustain disagreements required for a discrimination process and enable clustering, the dynamics must therefore be more complex.

Historically, researchers have introduced various behavioral and topological mechanisms to break the symmetry of global consensus. Bounded confidence models \cite{deffuantMixingBeliefsInteracting2000,hegselmannOpinionDynamicsBounded2002}, for instance, achieve clustering by assuming that agents sever communication ties with peers holding disparate views, resulting in a state-dependent topology that fragments into cohesive subgroups. Using this property, Mor\u{a}rescu and Girard explicitly applied opinion dynamics to graph community detection by introducing a decaying confidence bound. In their model, communities emerge as asymptotic consensus clusters among agents whose states converge at a prescribed rate \cite{morarescuOpinionDynamicsDecaying2011}. Alternatively, models incorporating antagonistic interactions on signed graphs \cite{altafiniConsensusProblemsNetworks2013} demonstrate that negative edges can drive a network toward a polarized, bipartite consensus, provided the underlying topology satisfies structural balance conditions. These frameworks capture opinion fragmentation but rely on dynamic topological changes or the assumption of negative edges. From the perspective of human social learning, however, these assumptions are often overly restrictive: social constraints frequently prevent the outright removal of edges, and individuals do not need to be antagonistic to maintain a disagreement.

However, transmitting exact, continuous internal states is unrealistic; internal opinions are private variables that are often not exactly quantified even by the agents themselves. Instead, agents observe the public actions or saturated signals of their peers. Distinguishing continuous internal beliefs from observable decisions shows that restricting the resolution of exchanged information generates persistent disagreement over a static, cooperative topology. This mechanism can be modeled through discontinuous quantized communication \cite{martinsContinuousOpinionsDiscrete2008,chowdhuryContinuousOpinionsDiscrete2016} or, more recently, through continuous sigmoidal feedback \cite{grayMultiagentDecisionMakingDynamics2018,baumannModelingEchoChambers2020,bizyaevaNonlinearOpinionDynamics2023,bizyaevaMultitopicBeliefFormation2025a}. This framework has been used to study polarization in opinion dynamics \cite{leonardNonlinearFeedbackDynamics2021a,baumannEmergencePolarizedIdeological2021}.

In the latter framework, introducing a high-gain nonlinearity into the interaction maps the continuous internal state to a saturated decision, forcing the system away from global consensus and into polarized states. The resulting dynamics act similarly to an Ising model \cite{martinsContinuousOpinionsDiscrete2008}: the nonlinearity creates a bistable potential that breaks the symmetry of the origin. Recent theoretical work also establishes how to tune the interaction nonlinearity to achieve a spectral threshold for clustering \cite{fontanMultiequilibriaAnalysisClass2018,bizyaevaNonlinearOpinionDynamics2023,couthuresGlobalSynchronizationMultiAgent2025}. This makes the social-learning approach a practical tool for community detection, as agents can now sustain disagreements necessary for clustering.

Decentralized community detection has been explored previously, most notably through Label Propagation Algorithms and their overlapping variants (e.g., COPRA, SLPA) \cite{raghavanLinearTimeAlgorithm2007,gregoryFindingOverlappingCommunities2010a,xieSLPAUncoveringOverlapping2011}. However, these approaches have a significant sociological and practical limitation. In label propagation frameworks, nodes explicitly exchange and update discrete community identifiers. This mechanism assumes that agents are willing and able to broadcast their group affiliations explicitly to all neighbors. In realistic social and distributed systems, however, communicating explicit community tags is often unnatural or may violate privacy constraints. Individuals do not interact by exchanging community IDs; rather, they share context-specific opinions or actions, and community boundaries are implicitly inferred from patterns of agreement and disagreement over time.

In this paper, we propose a decentralized framework for community detection driven entirely by nonlinear social learning dynamics. To our knowledge, this is the first application of continuous-time saturated opinion dynamics as a privacy-preserving clustering tool. The saturated nonlinear model from \cite{bizyaevaNonlinearOpinionDynamics2023} has recently been applied to identify Stochastic Block Models (SBMs) \cite{xingLearningCommunitiesEquilibria2024}, but previous efforts were restricted to two blocks and relied on a centralized perspective. In our approach, agents do not reveal their continuous internal opinions, nor do they communicate explicit community labels. Instead, they exchange only saturated, issue-specific actions. By tuning a simple interaction gain parameter beyond the spectral threshold, local agents can autonomously evaluate their structural affinity with their neighbors.

Our goal here is not to compete with highly optimized centralized algorithms like the Louvain or Leiden methods. Instead, we show that a fully decentralized process can still extract community structures of the with same properties. To achieve this we will use a nonlinear behavioral model for the node that is similar to social learning. Subsequently, to map continuous polarized states into robust discrete network communities while mitigating the sensitivity of the learning process to random initial conditions, we formulate and compare three distinct algorithmic implementations:
\begin{enumerate}
    \item \textbf{Recursive Neighbor Pruning (RNP)}, a baseline approach that recursively fragments the network based on the steady-state polarization of the agents from a single dynamic run;
    \item \textbf{Recursive Neighbor Pruning with Decaying Confidence (RNP-DC)}, a bounded confidence model where the network topology dynamically adapts during the transient phase of the opinion dynamics, severing ties based on a strictly shrinking tension threshold, effectively generalizing the foundational concepts introduced in \cite{morarescuOpinionDynamicsDecaying2011}.
    \item \textbf{Score-based Edge Reliability (SER)}, a static-graph approach that samples the basins of attraction over multiple independent realizations (simulating multiple, independent discussion topics) to construct a robust probabilistic map of intra-community links.
\end{enumerate}

The remainder of this paper is organized as follows. Section \ref{sec:problem_formulation} defines the network topology and the decentralized constraints of the problem. Section \ref{sec:opinion_dynamics} introduces the nonlinear interaction model, detailing its spectral threshold and bifurcation properties. Section \ref{sec:algorithms} formally presents the three community detection algorithms and discusses their respective analytical and computational tradeoffs. Finally, Section \ref{sec:numerical_validation} provides numerical validation and benchmarking against standard modularity metrics on synthetic and real-world networks, followed by concluding remarks in Section \ref{sec:conclusion}.

\textbf{Notation} In the following, we denote by $\R$ the set of real numbers. For a vector $\bm{x} \in \R^N$, we denote by $x_i$ the $i$-th component of $\bm{x}$. For a matrix $\bm{A} \in \R^{N\times N}$, we denote by $a_{ij}$ the element of $\bm{A}$ at the $i$-th row and $j$-th column. We denote by $\zero$ and $\one$ the vector of $\R^N$ with all components equal to $0$ and $1$, respectively. For a column vector $\bm{x} \in \R^N$, we denote by $\norm{\bm{x}} = (\bm{x}^\top \bm{x})^{\frac{1}{2}}$ the Euclidean norm of $\bm{x}$ and, for a positive definite matrix $\bm{D}$, by $\norm{\bm{x}}_{\bm{D}} = (\bm{x}^\top \bm{D} \bm{x})^{\frac{1}{2}}$ the norm induced by $\bm{D}$.
Moreover, we denote by $\mathrm{diag}(\bm{x}) \in \R^{N\times N}$ the diagonal matrix with diagonal elements given by the vector $\bm{x} \in \R^{N}$. For two vectors $\bm{x},\bm{y} \in \R^N$, $\bm{x} \leq \bm{y}$ means $x_i \leq y_j$ for all $i \in \{1,\dots,N\}$, if, in addition $\bm{x}\neq \bm{y}$, we note $\bm{x} < \bm{y}$. For a function $f: \mathcal{X} \to \mathcal{X}$, the set $\mathrm{Fix}(f) = \left\{ x \in \mathcal{X} \mid x =  f(x) \right\}$ contains the fixed points of $f$.

\section{Problem Formulation: Network Topology and Decentralized Clustering}
\label{sec:problem_formulation}

\subsection{Graph Structure and Network Interaction}

We consider a network of $N$ agents (or nodes) interacting over a graph $\mathcal{G} = (\mathcal{V}, \mathcal{E})$. Here, $\mathcal{V} = \{1, \dots, N\}$ is the set of vertices and $\mathcal{E} \subset \mathcal{V} \times \mathcal{V}$ is the set of edges. An edge $(i,j) \in \mathcal{E}$ indicates that agent $i$ receives information from agent $j$. The neighborhood of agent $i$, defined as the set of agents from which it receives information, is denoted by $\mathcal{N}_i = \{ j \in \mathcal{V} \mid (j,i) \in \mathcal{E} \}$, with its size (or in-degree) given by $d_i = |\mathcal{N}_i|$. A \emph{path} in $\Gcal$ is a finite sequence of edges $(i_1,i_2),(i_2,i_3),\dots,(i_p,i_{p+1})$ such that $(i_k,i_{k+1})\in \Ecal$  for all $k\in \{1,\dots,p\}$. Two vertices $i$, $j\in \Vcal$ are \emph{connected} in $\Gcal$ if there exists a path in $\Gcal$ joining $i$ and $j$ (i.e. $i_1=i$ and $j_p=j$). The graph $\Gcal$ is \emph{strongly connected} if any two vertices are connected. The standing assumptions of this interaction graph are as follows.

\begin{standassumption}\label{ass:graph}
    The graph $\Gcal$ is \textbf{time-invariant, symmetric, connected and simple} (i.e., it has no self-loops or multiple edges).
\end{standassumption}

The interaction topology is encoded by the \emph{adjacency matrix} $\bm{A} \in \mathbb{R}^{N\times N}$, where $a_{ij} = 1$ if $(j,i) \in \mathcal{E}$ and $a_{ij} = 0$ otherwise. We define the \emph{degree matrix} as $\bm{D} := \mathrm{diag}(d_1, \dots, d_N)$. From these, we construct the row-stochastic, \emph{random-walk normalized adjacency matrix}, $\bm{P} := \bm{D}^{-1}\bm{A}$. Note that by Assumption~\ref{ass:graph}, $d_i \geq 1$ for all $i$, meaning $\bm{D}$ is strictly positive definite and invertible.

\begin{remark}
    The main results leverage on the symmetry of the interaction graph. However, in practice, the algorithm proposed in Section~\ref{sec:algorithms} does not require this assumption.
\end{remark}

\subsection{The Decentralized Clustering Objective}

Having defined the interaction topology $\mathcal{G}$, we now formalize the community detection objective. A standard graph partitioning approach seeks to divide the network into a set of disjoint clusters $\mathcal{C} = \{\mathcal{V}_1, \dots, \mathcal{V}_k\}$ such that $\bigcup_{m=1}^k \mathcal{V}_m = \mathcal{V}$, where each $\mathcal{V}_m$ represents a densely internally connected subgraph with sparse external boundary edges.

However, complex networks rarely decompose into perfect isolated cliques or all-to-all components. They frequently contain \emph{frontier agents} (or boundary nodes) acting as bridges between multiple communities and experiencing competing influences. Centralized algorithms often arbitrarily force these boundary nodes into a single discrete group to optimize a global metric like modularity. Our objective is twofold: identify the densely connected core communities and isolate these frontier agents based on their structural ambiguity.

We require this structure to emerge dynamically as a macro-state of the system, subject to the following strict decentralized constraints:
\begin{enumerate}
    \item \textbf{Strict Locality:} An agent $i$ updates its state based solely on the observation of its immediate neighborhood $\mathcal{N}_i$. It has no knowledge of the global adjacency matrix $\bm{A}$, the network diameter, or its global position.
    \item \textbf{Restricted Communication (Privacy):} Agents are restricted to exchanging a single, time-varying scalar variable. They cannot transmit additional graph knowledge (e.g., their degree $d_i$, or neighbor lists), nor can they communicate explicit, discrete community labels.
    \item \textbf{Uniformity:} All agents must execute the exact same local update rule, with no designated leader nodes or centralized aggregation steps.
\end{enumerate}

Equipping the agents with a nonlinear communication protocol destabilizes global consensus. The network breaks symmetry: densely connected core subgroups tend to synchronize locally (or almost synchronize) and polarize from one another, whereas frontier agents, caught between competing influences, converge to neutral states or display high variability depending on the initialization. The algorithmic objective (Section \ref{sec:algorithms}) is to recover both the core communities and the frontier nodes from these state trajectories.

\section{Opinion Dynamics with Nonlinear Interaction}
\label{sec:opinion_dynamics}

\subsection{The Discrete-Time Model}

To satisfy the decentralized constraints outlined in Section \ref{sec:problem_formulation}, we map the structural properties of $\mathcal{G}$ to the steady states of a multi-agent dynamical system. Let $x_i(t) \in [-1,1]$ denote the private internal opinion of agent $i$ at discrete time step $t \in \mathbb{N}$. The interval $[-1,1]$ represents a normalized spectrum of belief regarding a specific topic, where $x_i \approx -1$ indicates strong opposition, $x_i \approx 1$ indicates strong support, and $x_i = 0$ represents strict neutrality.

Since internal opinions are private, agents interact only through their neighbors' observable actions. We model the transformation from private belief to public, saturated action through a common \textbf{signal function} $s: [-1,1] \to [-1,1]$.

Opinion updates occur over discrete communication steps and reflect a consensus-seeking behavior: an agent conforms to its local social neighborhood. This update is governed by the difference between the agent's current internal opinion and the average of the signals received from its neighbors. Formally, for any agent $i \in \Vcal$, this is modeled as the discrete-time social learning dynamics, given by:
\begin{equation}\label{eq:dynamic_single_agent}
    x_i(t+1) = x_i(t) + h \left( \frac{1}{d_i} \sum_{j \in \mathcal{N}_i} a_{ij} s(x_j(t)) - x_i(t) \right),
\end{equation}
where $h \in (0,1]$ is the step size, representing the learning rate or bounded inertia of the agents. The quantity in parentheses is a local disagreement term: when it is positive, the agent's state increases at the next time step; when it is negative, the state decreases.

To illustrate this sociologically, consider a social network discussing a controversial policy. An agent $j$ might hold only a weak internal preference for the policy ($x_j = 0.2$), having reservations while leaning positive. When acting publicly (voting, posting online, answering a survey), social pressure and limited communication bandwidth compel individuals to project a definitive, polarized stance. The signal function models this effect, transmitting a saturated endorsement (e.g., $s(0.2) \approx 0.8$). An adjacent agent $i$, who might initially be strictly neutral ($x_i = 0$), does not observe the nuanced hesitation of $j$ ($0.2$); it only observes the strong public endorsement ($0.8$). If the average of the public signals in the neighborhood of $i$ is highly positive, agent $i$ experiences a positive disagreement and adjusts its internal belief upward ($x_i(t+1) > x_i(t)$) to align with the perceived local norm.

The collective dynamics of all $N$ agents can be expressed in the compact vector form:
\begin{equation}\label{eq:dynamic_vector}
    \bm{x}(t+1) = \bm{x}(t) + h \big( \bm{D}^{-1}\bm{A} \bm{s}(\bm{x}(t)) - \bm{x}(t) \big) = \bm{x}(t) + h \big( \bm{P}\bm{s}(\bm{x}(t)) - \bm{x}(t) \big).
\end{equation}
where $\bm{x}(t) = (x_1(t),\dots,x_N(t))^\top \in \mathcal{X} := [-1,1]^N$ and $\bm{s}(\bm{x}) = (s(x_1),\dots,s(x_N))^\top$.

This discrete-time formulation corresponds precisely to the Euler discretization of the continuous-time multi-agent system
\begin{equation}\label{eq:dynamic_vector_continuous}
    \dot{\bm{x}} = \bm{D}^{-1}\bm{A} \bm{s}(\bm{x}) - \bm{x} = \bm{P}\bm{s}(\bm{x}) - \bm{x},
\end{equation}
which has been recently studied in the nonlinear consensus literature \cite{couthuresGlobalSynchronizationMultiAgent2025} for cooperative graph topologies. By shifting to a discrete-time framework, we align the theoretical physical model with the iterative algorithmic implementations required to compute the community structures on graphs.

\begin{remark}
    Using the random-walk normalized matrix $\bm{P} = \bm{D}^{-1}\bm{A}$ is a central modeling choice. Unlike dynamics governed by the unweighted graph Laplacian ($\bm{L} = \bm{D} - \bm{A}$), this formulation ensures the influence on an agent is an \emph{average} of received signals rather than a sum. The rate of an agent's response is independent of its degree, preventing highly connected hubs from dominating the transient dynamics. This reflects social environments where an individual is influenced by the prevailing proportion of opinions in their peer group rather than the sheer volume of connections.
\end{remark}

\begin{remark}\label{rem:heterogeneity}
    In realistic decentralized settings like human social networks, agents exhibit diverse behavioral traits. These may include agent-specific signal functions $s_i(x)$ \cite{fontanMultiequilibriaAnalysisClass2018}, heterogeneous step sizes, inherent stubbornness \cite{friedkinSocialInfluenceOpinions1990}, or simultaneous coupling across multiple topics \cite{bizyaevaNonlinearOpinionDynamics2023}. While incorporating such heterogeneity adds descriptive realism, it does not fundamentally alter the underlying macroscopic polarization mechanics, provided the individual signal functions satisfy the basic monotonicity and gain assumptions discussed below. For the engineering objective of community detection, these complex features are not required. Furthermore, assigning accurate individual behavioral parameters requires empirical metadata that is entirely unavailable in a privacy-preserving, decentralized context. We adopt the homogeneous, idealized model in \eqref{eq:dynamic_vector}. This ensures that any clustering behavior that emerges is driven strictly by the underlying network topology, rather than artificially introduced behavioral differences. However, if one seeks behavioral realism and possesses the necessary data to simulate such a system, a more flexible version of the model \eqref{eq:dynamic_vector} can be utilized, as presented in \cite{bizyaevaNonlinearOpinionDynamics2023, bizyaevaMultitopicBeliefFormation2025a}, where interconnected multi-dimensional opinion influence is captured.
\end{remark}

\subsection{Signal Function and Communication Model}

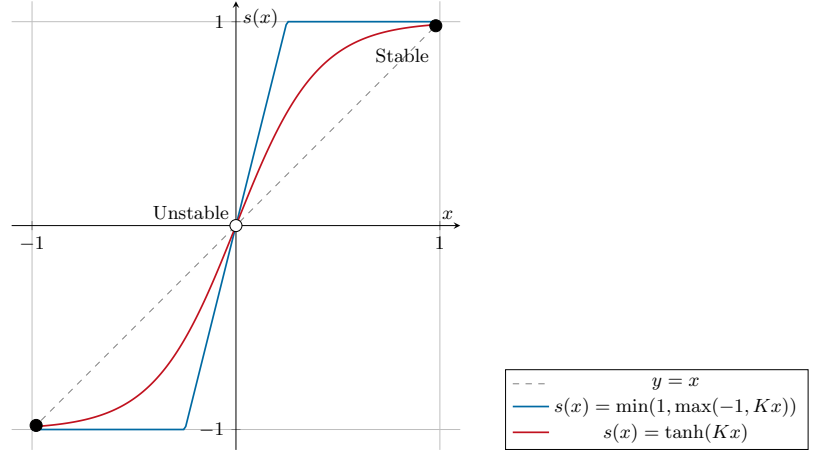
\begin{figure}[t]
    \centering
    \begin{tikzpicture}[scale=0.8]
        \begin{axis}[
                axis lines=middle,
                xlabel={$x$},
                ylabel={$s(x)$},
                xmin=-1.1, xmax=1.1,
                ymin=-1.1, ymax=1.1,
                xtick={-1, 0, 1},
                ytick={-1, 0, 1},
                width=9cm,
                height=9cm,
                grid=major,
                legend style={
                        at={(1.1, 0)},     
                        anchor=south west,  
                        name=mylegend       
                    },
            ]
            \addplot[domain=-1:1, samples=2, dashed, color=gray] {x};
            \addlegendentry{$y=x$}
            \addplot[domain=-1:1, samples=200, color=myblue, thick] {min(1,max(-1,4*x))};
            \addlegendentry{$s(x) = \min(1,\max(-1,Kx))$}

            \addplot[domain=-1:1, samples=200, color=myred, thick] {tanh(2.5*x)};
            \addlegendentry{$s(x) = \tanh(Kx)$}

            \node[circle,fill=white,draw=black,inner sep=2pt] at (axis cs:0,0) {};
            \node[circle,fill=black,draw=black,inner sep=2pt] at (axis cs:0.98,0.98) {};
            \node[circle,fill=black,draw=black,inner sep=2pt] at (axis cs:-0.98,-0.98) {};

            \node[anchor=south east] at (axis cs:0,0) {Unstable};
            \node[anchor=north west] at (axis cs:0.65,0.9) {Stable};
        \end{axis}

        \path let \p1 = ($(mylegend.south east) - (mylegend.south west)$) in ([xshift=-\x1]current axis.south west) -- (current axis.south west);

    \end{tikzpicture}
    \caption{\textbf{Example of nonlinear interaction functions.} Two examples of signal functions with an unstable origin: a smooth sigmoidal function $s(x) = \tanh(Kx)$ (red, $K=2.5$) and a saturated piecewise-linear function $s(x) = \min(1, \max(-1, Kx))$ (blue, $K=4$). The intersection with the identity line $y=x$ defines the fixed points. A steep slope at the origin ($K>1$) destabilizes the neutral consensus state $\bm{x}=\bm{0}$, forcing the system to evolve toward the stable clustered equilibria.}
    \label{fig:signal_function}
\end{figure}

The nonlinearity introduced by $s(x)$ determines the macroscopic behavior of the network. To sustain the disagreements necessary for community detection, we enforce the following assumption on the signal function.

\begin{standassumption}\label{ass:signal}
    The signal function $s: [-1,1] \to [-1,1]$ is non-decreasing and Lipschitz-continuous with constant $K > 0$ that is smooth at $0$ and $s'(0) = K$.
\end{standassumption}

The requirement that $s$ must be non-decreasing is central to our analysis. It captures a natural causality inherent to many social systems: a higher internal state must induce a non-lower public signal. A scenario where a positive shift in internal opinion causes a negative shift in public action would contradict cooperative local agreement-seeking behavior. Mathematically, this non-decreasing property guarantees that the dynamical system is monotone \cite{hirschChapter4Monotone2006}, preventing chaotic oscillations and ensuring convergence to steady states as proven in \cite{couthuresConsensusRobustClustering2025c}.

This formulation protects agent privacy by functioning as an imperfect information channel. If $s(x)$ were invertible, any neighboring agent could completely reverse-engineer the private internal state of agent $i$ by simply applying the inverse function, i.e., $x_i = s^{-1}(s(x_i))$. By utilizing a saturated, non-bijective signal function, we ensure that for any $x_1 \in [-1,1]$, there exists a distinctly different internal state $x_2 \in [-1,1]$ ($x_1 \neq x_2$) that maps to effectively the same public action, $s(x_1) = s(x_2)$.

Consequently, an observer seeing a saturated signal ($s(x) \approx 1$) knows the agent's internal opinion is positive, but cannot precisely differentiate between a moderately positive stance and an extreme one. The internal state remains hidden behind the public action. We therefore favor strongly saturated signal functions (blue curve, Fig.~\ref{fig:signal_function}) over strictly bijective ones (red line, Fig.~\ref{fig:signal_function}) to enforce structural polarization and strict data privacy.

As established in \cite{couthuresConsensusRobustClustering2025c}, the equilibria of the system are linked to the fixed points of the signal function, defined as the set $\mathrm{Fix}(s) = \{c \in [-1,1] \mid c = s(c)\}$. A fixed point represents a state where an agent communicates its exact internal opinion without any nonlinear distortion. When all agents reach a common state $c \in \mathrm{Fix}(s)$, the system achieves a consensus state. To characterize how $s(x)$ shapes the dynamics around these equilibria, we classify the distortion and local stability based on the term $s(x)-x$.

\begin{definition}[{\cite{couthuresConsensusRobustClustering2025c}}]\label{def:fixed_point_stability}
    A signal function $s$ is said to be:
    \begin{itemize}
        \item an \emph{underestimation} if $x(s(x) - x) \leq 0$ for all $x \in \left[-1,1\right]$.
        \item an \emph{overestimation} if $x(s(x) - x) \geq 0$ for all $x \in \left[-1,1\right]$.
    \end{itemize}
    Furthermore, a fixed point $c \in \mathrm{Fix}(s)$ is classified based on its one-sided stability:
    \begin{itemize}
        \item It is \emph{left-stable} (resp. \emph{right-stable}) if there exists a neighborhood $I \subseteq [-1,1]$ of $c$ such that $(x-c)(s(x)-x)\leq 0$ for all $x \in I$ with $x<c$ (resp. $x>c$).
        \item It is \emph{left-unstable} (resp. \emph{right-unstable}) if the strict inequality $(x-c)(s(x)-x)> 0$ holds under the same conditions.
    \end{itemize}
    Based on these properties, we define a fixed point as \emph{stable} if it is both left- and right-stable, \emph{unstable} if it is both left- and right-unstable, and \emph{semi-stable} otherwise.
\end{definition}

These classifications carry direct sociological and dynamical interpretations. Sociologically, \emph{underestimation} implies agents moderate their beliefs during communication, drawing the network toward a central consensus. Conversely, \emph{overestimation} amplifies internal states, driving polarization. Dynamically, the term $s(x)-x$ acts as a restoring force around \emph{stable} fixed points and a repelling force around \emph{unstable} or \emph{semi-stable} points.

As shown in \cite{couthuresConsensusRobustClustering2025c}, this one-dimensional classification dictates the global behavior of the network. We partition the fixed points into a stable set, $\mathrm{Fix}^{\bullet}(s)$, and an unstable/semi-stable set, $\mathrm{Fix}^{\circ}(s)$. While global consensus requires convergence to a point in $\mathrm{Fix}^{\bullet}(s)$, sustained disagreement requires at least one unstable point in $\mathrm{Fix}^{\circ}(s)$. This unstable point acts as a structural splitter, forcing agents on either side to diverge toward opposing, stable equilibria.

\section{Model Properties: Structural Splitting and Spectral Clustering}
\label{sec:structural_splitting}

While an unstable fixed point triggers structural splitting, whether the network fragments depends on a balance between local node dynamics and global graph topology. This section summarizes the mathematical properties governing this phase transition to justify the algorithmic methods proposed in Section~\ref{sec:algorithms}. Formal proofs for the results in this section are established in our companion paper \cite{couthuresConsensusRobustClustering2025c}.

\subsection{The Spectral Threshold for Consensus}

To quantify the structural connectivity of the network, we examine the spectrum of the row-stochastic matrix $\bm{P} = \bm{D}^{-1}\bm{A}$. By the Perron-Frobenius theorem, for a strongly connected graph, the largest eigenvalue is $\lambda_N = 1$, associated with the consensus eigenvector $\bm{1}$. We order the real eigenvalues as $-1 \leq \lambda_1 \leq \dots \leq \lambda_{N-1} < 1$, where $1 - \lambda_{N-1}$ is the algebraic connectivity of the random-walk normalized Laplacian. A value of $\lambda_{N-1}$ close to $1$ indicates a network with severe topological bottlenecks (sparse cuts between dense communities), whereas a value close to $0$ or negative indicates a highly cohesive, densely connected graph.

The global averaging effect of this topology competes directly with the local polarizing force of the signal function, parameterized by the origin gain $K = s'(0)$. This competition creates a sharp threshold for consensus.

\begin{theorem}[\cite{couthuresConsensusRobustClustering2025c}] \label{thm:sharp_threshold}
    If the gain at the origin satisfies $K\lambda_{N-1} < 1$, then all equilibria of the system \eqref{eq:dynamic_vector_continuous} are strictly consensus states, belonging to $\mathrm{Span}(\bm{1})$. Furthermore, this consensus manifold is globally attractive.
\end{theorem}

Theorem \ref{thm:sharp_threshold} dictates that if the network is sufficiently dense relative to the nonlinearity, structural splitting is impossible. The topological averaging suppresses the local instability at the origin, pulling all agents into a unified consensus.

\begin{remark}
    This result provides a direct connection to the standard linear DeGroot model \cite{degrootReachingConsensus1974}. In linear consensus, $s(x) = x$, which corresponds to a constant gain of $K=1$. Since $\lambda_{N-1} < 1$ is guaranteed for any connected graph, the condition $1 \cdot \lambda_{N-1} < 1$ is trivially satisfied. This formally explains why linear dynamics inevitably erase structural diversity and cannot be used for clustering.
\end{remark}

To utilize this system for community detection, we must tune the interaction gain to violate this threshold ($K\lambda_{N-1} > 1$). This requires the network to possess topological bottlenecks (sparse cuts between dense communities), which corresponds to $\lambda_{N-1}$ being strictly positive and close to $1$. If the graph is instead highly cohesive—such as a complete or complete bipartite graph—the second-largest eigenvalue satisfies $\lambda_{N-1} \leq 0$. In these specific cases, $K\lambda_{N-1} < 1$ holds for any $K > 0$, rendering the network structurally immune to polarization and guaranteeing global consensus regardless of the chosen nonlinearity.

\subsection{Symmetry Breaking and the Attractor Landscape}
\label{sec:symmetry_breaking}

Once this spectral threshold is violated, the neutral consensus state $\bm{x} = \bm{0}$ undergoes a pitchfork bifurcation and loses stability \cite{grayMultiagentDecisionMakingDynamics2018,bizyaevaNonlinearOpinionDynamics2023}.

The mechanism by which the network breaks symmetry connects our dynamical model to classical spectral clustering. Consider the linearization of the dynamics \eqref{eq:dynamic_vector_continuous} around the unstable neutral state $\bm{x} = \bm{0}$. Since $s(x) \approx Kx$ for small $x$, the local dynamics is governed by:
\begin{equation} \label{eq:linearization}
    \dot{\bm{x}} \approx (K\bm{D}^{-1}\bm{A} - \bm{I})\bm{x}.
\end{equation}
The eigenvalues of the Jacobian $\bm{J}(\bm{0}) = K\bm{P} - \bm{I}$ are given by $\nu_i = K\lambda_i - 1$. Because $K > 1/\lambda_{N-1}$, the eigenvalue $\nu_{N-1}$ becomes strictly positive. Any small random initialization $\bm{x}(0)$ near the origin is exponentially amplified along the direction of the corresponding right eigenvector, $\bm{v}_{N-1}$.

The eigenvector $\bm{v}_{N-1}$ corresponds in sign to the Fiedler vector of the normalized graph Laplacian. In standard spectral graph theory \cite{spielmanSpectralPartitioningWorks2007}, the signs of the Fiedler vector are used to partition a graph along its sparsest cut. For initial conditions close to this $\bm{v}_{N-1}$, the early transient phase of our multi-agent system splits the network into two factions aligned with the structural boundaries of the graph.

As the state amplitudes grow, the nonlinearity takes effect. The state space evolves into an attractor landscape with multiple stable polarized equilibria and their respective basins of attraction. Because the final steady-state configuration depends on the initial state $\bm{x}(0)$, a uniform random initialization samples this landscape, allowing the network to converge to different local basins highlighting varying structural cuts. This motivates our choice of signal functions with a unique unstable fixed point at the origin (as in Fig.~\ref{fig:signal_function}), simplifying the clustering process into a binary separation around $0$.

\subsection{State Confinement and Robustness of Core Communities}
\label{sec:cluster_stability}

While spectral bisection explains above the initial transient divergence, the final stage of clustering relies on the steady-state confinement and robustness of the agents' opinions. To classify nodes into discrete communities, once a densely connected core polarizes, its members must stabilize and reject conflicting influences from external frontier agents, preventing them from switching back across the neutral origin.

We formalize this robustness by analyzing a tightly knit subset of agents, representing a candidate core community $\Vcal' \subseteq \Vcal$, as an open dynamical sub-system subjected to external perturbations. For any agent $i \in \Vcal'$, we partition its degree $d_i$ into its internal degree $d_i^I = |\mathcal{N}_i \cap \Vcal'|$ and its external degree $d_i^E = |\mathcal{N}_i \setminus \Vcal'|$.

Assume the agents inside $\Vcal'$ are being drawn toward a common, stable fixed point of the signal function, $c \in \mathrm{Fix}^{\bullet}(s)$. Because $c$ is a stable equilibrium, the local gain around $c$ is typically small (e.g., for a saturated piecewise-linear function, the derivative at $c=1$ is exactly zero). We capture this mathematically by defining a local Lipschitz constant $L_c \geq 0$ within a neighborhood $\mathcal{I}_c$ of $c$, such that $|s(x) - s(c)| \leq L_c |x - c|$ for all $x \in \mathcal{I}_c$.

The structural coherence of an agent $i$ with respect to this community is defined by the ratio of its internal to total connections. The following proposition establishes that, provided the community's state remains within this local neighborhood, its maximum deviation from $c$ is Input-to-State Stable (ISS) in discrete time.

\begin{proposition}[Robustness of Core Communities and Internal Coherence]\label{prop:iss_cell_dev}
    Let $\Vcal' \subseteq \Vcal$ be a sub-network and $c \in \mathrm{Fix}^{\bullet}(s)$. Assume there exists an interval $\mathcal{I}_c$ containing $c$ such that $s$ is Lipschitz-continuous with constant $L_c \geq 0$ on $\mathcal{I}_c$. For all agents $i \in \Vcal'$ at time step $t$, define:
    \begin{itemize}
        \item the \textbf{internal coherence ratio} as $\mu_i = d_i^I/d_i \in [0,1]$;
        \item the \textbf{deviation from consensus} as $e_i(t) = x_i(t) - c$;
        \item the \textbf{maximal external perturbation} as $u_i(t) = \max_{j \in \Ncal_i \setminus \Vcal'} | s(x_j(t)) - c |$, with magnitude $\|u(t)\|_{\infty} = \max_{i \in \Vcal'} u_i(t)$;
        \item the \textbf{internal stability margin} as $\eta_i = 1 - L_c \mu_i$, with $\eta = \min_{i \in \Vcal'} \eta_i$;
        \item the \textbf{effective perturbation gain} as $\gamma_i = (1 - \mu_i)/\eta_i$, with $\gamma_{\max} = \max_{i \in \Vcal'} \gamma_i$.
    \end{itemize}

    Assume the network topology and the local gain $L_c$ are such that $\eta > 0$. For any step size $h \in (0, 1]$, if $x_i(\tau) \in \mathcal{I}_c$ for all $i \in \Vcal'$ and all $0 \leq \tau \leq t$, let $\delta(t) = \max_{i \in \Vcal'} |e_i(t)|$ be the maximum absolute deviation within the core. Then, $\delta(t)$ satisfies the discrete-time ISS bound:
    \begin{equation*}
        \delta(t) \le (1 - h\eta)^t \delta(0) + \gamma_{\max} \max_{0 \le \tau < t} \|u(\tau)\|_{\infty}.
    \end{equation*}
    Consequently, provided the trajectories remain in $\mathcal{I}_c$, the ultimate asymptotic bound on the community's deviation is:
    \begin{equation*}
        \limsup_{t \to \infty} \delta(t) \le \gamma_{\max} \limsup_{t \to \infty} \|u(t)\|_{\infty}.
    \end{equation*}
\end{proposition}

With Proposition~\ref{prop:iss_cell_dev}, we can interpret a core community as an elastic body under continuous external stress. Here, the internal coherence ratio ($\mu_i$) defines the network's structural stiffness. Because of the ISS condition, the community behaves like a damped mass-spring system: transient disagreements decay exponentially, and agents can only be pushed so far from their local consensus. We bound this maximum displacement by taking the maximal external force ($u(t)$) and scaling it by the effective perturbation gain ($\gamma_{\max}$). The denser the internal topology, the ``stiffer'' the community: meaning it will by itself absorbs external noise and stops frontier nodes from dragging the core out of its stable state.

\begin{remark}[Impact of Internal Coherence on Robustness]
    The parameter $\gamma_i$ links an agent's topological placement and dynamic robustness. For typical stable equilibria where $L_c < 1$, the effective perturbation gain $\gamma_i(\mu_i) = (1-\mu_i) / (1-L_c \mu_i)$ is a strictly decreasing function of the internal coherence $\mu_i = d^I_i/d_i$. Taking the derivative yields $\gamma_i'(y) = (L_c - 1) / (1-L_c y)^2 < 0$ for any $y \in [0,1]$. This guarantees that agents with denser internal connections (higher $\mu_i$) exhibit a tighter asymptotic bound against external perturbations, validating the intuition that core community members are intrinsically more resilient to frontier influences.
\end{remark}

However, while ISS provides an ultimate bound on the influence of external nodes, community detection algorithms mapping states to discrete partitions rely on strict state confinement. Because the system is cooperative (monotone) and the discrete update operates as a strict convex combination for $h \in (0,1]$, the state space $\mathcal{X} = [-1,1]^N$ is partitioned into forward-invariant hypercubes \cite{couthuresConsensusRobustClustering2025c}. To confidently classify nodes, we must guarantee that once a dense core brings its members sufficiently close to a stable fixed point $c$, no external topology can pull them back across the neutral origin.

\begin{proposition}[Forward Invariance of Community States]\label{prop:cell_deviation_bound}
    Let $\Vcal' \subseteq \Vcal$ be a candidate community and $c \in \mathrm{Fix}^{\bullet}(s)$. For a given tolerance $\varepsilon > 0$, define the \textbf{maximum external perturbation amplitude} as:
    \begin{align*}
        p(c,\varepsilon) & = \max_{y \in \{-1,1\}, \, \sigma \in \{-1,1\}} \big| s(y) - (c + \sigma \varepsilon)\big|,
    \end{align*}
    and the \textbf{minimal restoring force amplitude} generated by the local consensus as:
    \begin{align*}
        r(c,\varepsilon) & = \min\big((c+\varepsilon) - s(c+\varepsilon), s(c-\varepsilon) - (c-\varepsilon) \big).
    \end{align*}

    If, for all agents $i \in \Vcal'$, the network topology satisfies the strict boundary condition:
    \begin{equation} \label{eq:condition_invariant}
        d^E_i p(c,\varepsilon) \leq d^I_i r(c,\varepsilon),
    \end{equation}
    and the step size satisfies $h \in (0,1]$, then the hypercube
    $\mathcal{H}(c, \varepsilon) = \{ \bm{x} \in \mathcal{X} \mid \forall i \in \Vcal', x_i \in [c-\varepsilon, c+\varepsilon] \}$
    is forward invariant under the discrete-time dynamics \eqref{eq:dynamic_vector}.
\end{proposition}

While Proposition~\ref{prop:iss_cell_dev} gives us an elastic bound on perturbations, Proposition~\ref{prop:cell_deviation_bound} establishes the frontier of the attraction basin. The inequality $d_i^E p(c, \varepsilon) \le d_i^I r(c, \varepsilon)$ formulates a worst-case force balance. The left side captures the worst case maximum outward pull from external connections; the right side defines the weakest possible internal restoring force drawing the agent back to $c$. When the internal force dominates, the hypercube $\mathcal{H}(c, \varepsilon)$ becomes a forward-invariant basin of attraction. Once a dense cluster of agents falls into this basin, the rest of the network lacks the topological leverage to drag them back across the unstable origin. At that point, the discrete community boundaries are permanently locked.

Together, these propositions formalize the robustness of discrete communities in our continuous-state system. Once a densely connected sub-network $\Vcal'$ possesses a sufficiently high ratio of internal edges to satisfy Proposition \ref{prop:cell_deviation_bound}, its state becomes permanently locked within an invariant hypercube by the rest of the network. This confinement proves that structural boundaries in the graph translate to deterministic boundaries in the state space.

\subsection{Stochastic Sensitivity and Frontier Nodes}
\label{sec:frontier_nodes}

While Proposition \ref{prop:cell_deviation_bound} guarantees the permanent confinement of densely connected core communities once they enter an invariant hypercube, networks rarely consist exclusively of isolated cliques. They frequently contain \emph{frontier agents}: nodes situated on the topological boundaries between multiple communities.

Because these frontier agents lack the high internal degree ratio ($d_i^I / d_i \approx 1$) required for state confinement, they experience conflicting, zero-sum influences from opposing dense groups. They are structurally prevented from fully saturating, and their asymptotic behavior becomes highly sensitive to initial conditions.

Depending on the specific random initialization $\bm{x}(0)$, a frontier node might be pulled into the positive invariant hypercube in one dynamical realization, dragged into the negative hypercube in another, or remain trapped near the unstable origin ($\bm{x}_i \approx 0$).

Evaluating the dynamics over multiple independent iterations and measuring the sensitivity of the agents' steady-state agreements separates deterministic structural bounds from the stochastic noise of initial conditions. Agents deep within a core community consistently fall into the same basin of attraction and occupy the same invariant hypercube, preserving their topological link. Conversely, the high variability of frontier nodes exposes their structural ambiguity. This principle forms the theoretical foundation for the statistical filtering methods detailed in Section \ref{sec:algorithms}.

\section{Algorithmic Implementations}\label{sec:algorithms}

\begin{figure}[t]
    \centering
    \begin{subfigure}[t]{0.32\textwidth}
        \includegraphics{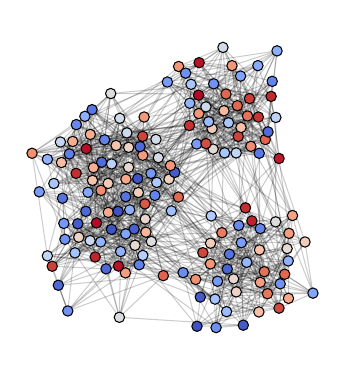}
        \caption{Step 1: Initialization and dynamics.}
        \label{fig:recursive_pruning_initial}
    \end{subfigure}
    \begin{subfigure}[t]{0.32\textwidth}
        \includegraphics{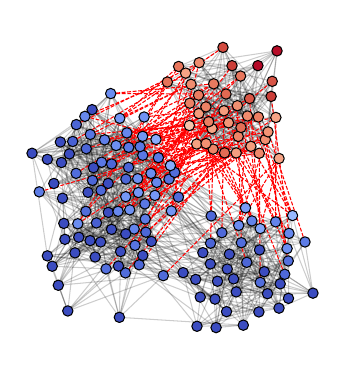}
        \caption{Step 2: Pruning of edges with opposing signs.}
        \label{fig:recursive_pruning_pruned}
    \end{subfigure}
    \begin{subfigure}[t]{0.32\textwidth}
        \includegraphics{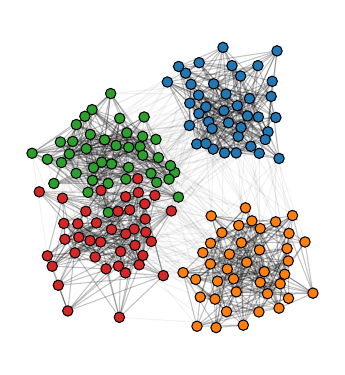}
        \caption{Final result: the graph is completely partitioned into 4 communities.}
        \label{fig:recursive_pruning_final}
    \end{subfigure}
    \caption{\textbf{Visualization of the steps in the Recursive Neighbor Pruning (RNP) Algorithm.} (a) Step 1: Agents are initialized with random continuous opinions $x_i \sim \mathcal{U}(-1, 1)$ across the fully connected topology. (b) Step 2: Once the continuous dynamics reach a polarized steady state, edges connecting agents with opposing signs ($s_i \cdot s_j < 0$) are identified as structural bottlenecks and severed (dashed red lines). (c) Final network partition obtained by recursively applying this dynamic bisection to the resulting subgraphs until no further polarization occurs. The network is a SBM with $N=160$ nodes (four blocks of $40$), intra-community probability $p_{\text{in}} = 0.30$, and inter-community probability $p_{\text{out}} \le 0.05$. The dynamics were simulated with interaction gain $K=5$ and step size $h=0.1$.}
    \label{fig:recursive_pruning}
\end{figure}

The theoretical framework in Section \ref{sec:structural_splitting} shows that tuning the interaction gain $K > 1/\lambda_{N-1}$ causes the states of the agents to polarize along the structural boundaries of the graph. Using this continuous dynamics for community detection requires a discrete edge-cutting logic. Because polarized equilibria depend on the initial state $\bm{x}(0)$, the algorithms must also account for stochastic sensitivity.

We propose three algorithms:
\begin{itemize}[nosep]
    \item Recursive Neighbor Pruning (RNP)
    \item Recursive Neighbor Pruning with Decaying Confidence (RNP-DC)
    \item Score-based Edge Reliability (SER)
\end{itemize}
to solve this mapping problem.

\subsection{The Baseline: Recursive Neighbor Pruning (RNP)}

The most direct translation of our discrete-time model into a clustering algorithm is a recursive bisection approach, detailed in Algorithm \ref{algo:RNP} and illustrated in Fig.~\ref{fig:recursive_pruning}. The protocol initializes agents' states randomly around the origin (Fig.~\ref{fig:recursive_pruning_initial}) and executes the nonlinear dynamics \eqref{eq:dynamic_single_agent} until the system reaches steady-state equilibrium ($\max_i |\Delta x_i| < 10^{-6}$).

Because the system is driven into invariant hypercubes associated with stable fixed points (Section \ref{sec:cluster_stability}), the network polarizes into two subsets: agents with positive saturated signals and agents with negative ones. The algorithm evaluates the network structure. If adjacent agents $i$ and $j$ hold opposing signs ($s_i \cdot s_j < 0$), they mathematically occupy disjoint invariant hypercubes. The edge $(i,j)$ represents a topological bottleneck and is severed (Fig.~\ref{fig:recursive_pruning_pruned}). This pruning process is applied recursively to the disconnected subgraphs until no further polarization occurs, revealing the discrete communities (Fig.~\ref{fig:recursive_pruning_final}).

\begin{algorithm}[t]
    \SetAlgoLined
    \DontPrintSemicolon
    \SetKwInOut{Input}{Input}
    \SetKwInOut{Output}{Output}
    \SetKwInOut{Parameter}{Parameter}
    \SetKwProg{Fn}{Function}{}{end}
    \SetKwProg{Break}{break}{}

    \Input{Graph $\mathcal{G} = (\mathcal{V}, \mathcal{E})$}
    \Output{Graph $\mathcal{G}_{\mathrm{new}} = (\mathcal{V}, \mathcal{E}_{\mathrm{new}})$ with no edges connecting agents with opposing signs}
    \Parameter{Signal function $s$, Time step $h$, Max steps $T_{\max}$, Max retries $N_{\mathrm{retry}}$}

    \BlankLine
    \Fn{\textsc{RecursivePruning}($\mathcal{G}$)}{
        \For(\tcp*[f]{Run multiple attempts to find a stable graph}){$\mathrm{attempt} = 1$ \KwTo $N_{\mathrm{retry}}$}{

            \lForEach(\tcp*[f]{Initialize decentralized opinions}){$\mathrm{agent} \ i \in \mathcal{V}$}{
                $x_i \gets \mathcal{U}(-1, 1)$
            }

            \For(\tcp*[f]{Run dynamics up to steady state or maximum steps}){$t = 0$ \KwTo $T_{\max}$}{
                \lForEach(\tcp*[f]{Broadcast public signal}){$\mathrm{agent} \ i \in \mathcal{V}$}{
                    $s_i \gets s(x_i)$
                }
                \ForEach{$\mathrm{agent} \ i \in \mathcal{V}$}{
                    $\Delta x_i = \frac{1}{d_i} \sum_{j \in \mathcal{N}_i} \left( s_j - x_i \right)$ \tcp*{Individual opinion delta from neighbors' signals}
                    $x_i \gets x_i + h \Delta x_i$ \tcp*{Update individual opinion}
                }

                \lIf(\tcp*[f]{Local early stopping check}){$\max_{i \in \mathcal{V}} |\Delta x_i| < 10^{-6}$}{
                    \textbf{break}
                }
            }
            $\mathcal{E}_{\mathrm{new}} \gets \mathcal{E}$ \tcp*{Initialize new graph}
            $\mathrm{pruned} \gets \text{False}$ \tcp*{Initialize pruning flag}
            \ForEach{$\mathrm{agent} \ i \in \mathcal{V}$}{
                \ForEach{$\mathrm{neighbor} \ j \in \mathcal{N}_i$}{
                    \If(\tcp*[f]{Check for opposing signs}){$s_i \cdot s_j < 0$}{
                        $\mathcal{E}_{\mathrm{new}} \gets \mathcal{E}_{\mathrm{new}} \setminus \{(i,j)\}$ \tcp*{Remove edge from new graph}
                        $\mathrm{pruned} \gets \text{True}$ \tcp*{Set pruning flag}
                    }
                }
            }
            $\mathcal{G}_{\mathrm{new}} \gets (\mathcal{V}, \mathcal{E}_{\mathrm{new}})$\;
            \If(\tcp*[f]{Check if pruning occurred}){$\mathrm{pruned}$}{
                \Return $\textsc{RecursivePruning}(\mathcal{G}_{\mathrm{new}})$ \tcp*{Recurse on new graph}
            }
        }
        \Return $\mathcal{G}_{\mathrm{new}}$ \tcp*{Stable graph found after maximum attempts}
    }
    \caption{Recursive Neighbor Pruning (RNP)}
    \label{algo:RNP}
\end{algorithm}

While RNP provides an intuitive distributed implementation of spectral bisection, it remains highly sensitive to initial conditions. A single random initialization may inadvertently lead to a global consensus, resulting in no structural cut. To mitigate this effect, RNP requires multiple retries ($N_{\mathrm{retry}}$) per subgraph to reliably confirm whether a community is truly indivisible, a requirement that ultimately impacts computational efficiency.

\begin{algorithm}[t]
    \SetAlgoLined
    \DontPrintSemicolon
    \SetKwInOut{Input}{Input}
    \SetKwInOut{Output}{Output}
    \SetKwInOut{Parameter}{Parameter}
    \SetKwProg{Fn}{Function}{}{end}

    \Input{Adjacency matrix $\bm{A} \in \mathbb{R}^{n \times n}$, Node indices $\mathcal{V} \subseteq \{1, \dots, N\}$}
    \Output{Set of disjoint communities $\mathcal{C} = \{\mathcal{V}_1, \dots, \mathcal{V}_k\}$}
    \Parameter{Signal function $s$, Step size $h$, Max steps $T_{\max}$, Max retries $N_{\mathrm{retry}}$}

    \BlankLine
    \Fn{\textsc{EfficientCentralizedRNP}($\bm{A}$, $\mathcal{V}$)}{
        \lIf(\tcp*[f]{Base case: core community reached}){$|\mathcal{V}| \leq 3$}{
            \Return $\{\mathcal{V}\}$
        }

        $\bm{D} \gets \mathrm{diag}(\bm{A}\bm{1})$\;
        $\bm{P} \gets \bm{D}^{-1}\bm{A}$ \tcp*{Precompute row-stochastic matrix}
        $\mathrm{split\_found} \gets \text{False}$\;

        \For{$\mathrm{attempt} = 1$ \KwTo $N_{\mathrm{retry}}$}{
            $\bm{x} \gets \mathcal{U}(-1, 1)^{|\mathcal{V}|}$ \tcp*{Uniform random initialization}

            \For{$t = 1$ \KwTo $T_{\max}$}{
                $\Delta \bm{x} \gets \bm{P}\bm{s}(\bm{x}) - \bm{x}$\;
                $\bm{x} \gets \bm{x} + h\Delta \bm{x}$ \tcp*{Vectorized state update}

                \lIf(\tcp*[f]{Early stopping on steady state}){$\|\Delta \bm{x}\|_{\infty} < 10^{-6}$}{
                    \textbf{break}
                }
            }

            $\mathcal{V}^+, \mathcal{V}^- \gets \{ i \in \mathcal{V} \mid x_i \geq 0 \}, \{ i \in \mathcal{V} \mid x_i < 0 \}$ \tcp*{Partition nodes by sign}

            \If(\tcp*[f]{Valid structural bisection achieved}){$|\mathcal{V}^+| > 0$ \textbf{and} $|\mathcal{V}^-| > 0$}{
                $\mathrm{split\_found} \gets \text{True}$\;
                \textbf{break}
            }
        }

        \eIf{$\mathrm{split\_found}$}{
            $\bm{A}^+, \bm{A}^- \gets \bm{A}[\mathcal{V}^+, \mathcal{V}^+], \bm{A}[\mathcal{V}^-, \mathcal{V}^-]$ \tcp*{Extract sub-adjacency matrices}
            $\mathcal{C}^+, \mathcal{C}^- \gets \textsc{EfficientCentralizedRNP}(\bm{A}^+, \mathcal{V}^+), \textsc{EfficientCentralizedRNP}(\bm{A}^-, \mathcal{V}^-)$\;
            \Return $\mathcal{C}^+ \cup \mathcal{C}^-$
        }{
            \Return $\{\mathcal{V}\}$
        }
    }
    \caption{Efficient Centralized Implementation of RNP}
    \label{algo:RNP_Centralized}
\end{algorithm}

Although Algorithm \ref{algo:RNP} formulates the decentralized execution of the dynamics from the perspective of individual agents, directly simulating this local protocol centrally is inefficient. From a centralized computational standpoint, RNP is equivalent to a recursive bisection algorithm operating on isolated sub-adjacency matrices. As detailed in Algorithm \ref{algo:RNP_Centralized}, an efficient centralized implementation vectorizes the state update and immediately extracts independent subgraphs upon any polarization event. The dynamics is then recursively applied only to these smaller, dynamically decoupled components.

When implementing Algorithm \ref{algo:RNP_Centralized}, it may appear computationally advantageous to abandon random initial conditions for a spectral initialization---setting the initial state $\bm{x}(0)$ proportional to the Fiedler eigenvector of the normalized Laplacian. Because the Fiedler vector aligns with the unstable manifold of the origin (Section \ref{sec:structural_splitting}), this guarantees rapid, deterministic convergence to a polarized state, reducing the continuous dynamics to a nonlinear rounding scheme for standard spectral bisection. We still rely on random initialization ($\bm{x}(0) \sim \mathcal{U}(-1, 1)^N$) for three reasons.

First, an eigenvector initialization requires global topological knowledge, violating the decentralized constraints from Section \ref{sec:problem_formulation}. Second, for strict spectral bisection, classical linear algebra solvers \cite{fortunatoCommunityDetectionGraphs2010} are computationally superior to numerically integrating a differential equation. Finally, a Fiedler initialization forces a rigid binary cut. In contrast, the nonlinear system \eqref{eq:dynamic_vector} possesses a complex attractor landscape. Uniform random initialization samples this space, allowing the network to converge to various local basins. This stochastic exploration frequently fractures the network into more than two connected components upon edge pruning, identifying structural boundaries classical algorithms miss.

\subsection{Bounded Confidence Adaptation: Recursive Neighbor Pruning with Decaying Confidence (RNP-DC)}

\begin{algorithm}[t]
    \SetAlgoLined
    \DontPrintSemicolon
    \SetKwInOut{Input}{Input}\SetKwInOut{Output}{Output}
    \SetKwInOut{Parameter}{Parameter}
    \SetKwHangingKw{Break}{break}
    \SetKwProg{Fn}{Function}{}{end}

    \Input{Graph $\mathcal{G} = (\mathcal{V}, \mathcal{E})$}
    \Output{Graph $\mathcal{G}_{\mathrm{new}} = (\mathcal{V}, \mathcal{E}_{\mathrm{new}})$ with no edges between communities}
    \Parameter{Signal function $s$, Time step $h$, Max steps $T_{\max}$, Number of retrials $N_{\mathrm{retry}}$, Decay rate $\rho$, Initial radius $R$}

    \Fn{\textsc{RecursiveDecayingConfidencePruning}($\mathcal{G}$)}{

        \For{$\mathrm{attempt} = 1$ \KwTo $N_{\mathrm{retry}}$}{
            \lForEach(\tcp*[f]{Initialize decentralized opinions}){$\mathrm{agent} \ i \in \mathcal{V}$}{
                $x_i \gets \mathcal{U}(-1, 1)$
            }
            $\mathcal{E}_{\mathrm{new}} \gets \mathcal{E}$ \tcp*{Initialize new graph}
            $\mathrm{pruned} \gets \text{False}$ \tcp*{Initialize pruning flag}

            \For(\tcp*[f]{Bounded confidence transient dynamics}){$t = 1$ \KwTo $T_{\max}$}{
                $r_t \gets R \cdot \rho^t$ \tcp*{Shrinking tension threshold}
                \lForEach(\tcp*[f]{Broadcast public signal}){$\mathrm{agent} \ i \in \mathcal{V}$}{
                    $s_i \gets s(x_i)$
                }
                \ForEach(\tcp*[f]{Topology adaptation phase}){$\mathrm{edge} \ (i,j) \in \mathcal{E}$}{
                    \If(\tcp*[f]{Check if edge is severed}){$|s_i - s_j| > r_t$}{
                        $\mathcal{E}_{\mathrm{new}} \gets \mathcal{E}_{\mathrm{new}} \setminus \{(i,j)\}$ \tcp*{Remove edge from new graph}
                        $\mathrm{pruned} \gets \text{True}$ \tcp*{Set pruning flag}
                    }
                }

                \ForEach(\tcp*[f]{State update phase}){$\mathrm{agent} \ i \in \mathcal{V}$}{
                    $\mathcal{N}_i \gets \{ j \in \mathcal{V} \mid (j,i) \in \mathcal{E}_{\mathrm{new}} \}$ \tcp*{Update neighborhood of agent $i$ based on new graph}
                    $d_i \gets |\mathcal{N}_i|$ \tcp*{Update degree of agent $i$ based on new graph}
                    \lIf(\tcp*[f]{Check if agent is isolated}){$d_i = 0$}{ $\Delta x_i \gets 0$ }
                    \lElse{
                        $\Delta x_i = \frac{1}{d_i} \sum_{j \in \mathcal{N}_i} \left( s_j - x_i \right)$ \tcp*{Individual opinion delta from neighbors' signals}
                        $x_i \gets x_i + h \Delta x_i$ \tcp*{Update individual opinion}
                    }
                }

                \lIf(\tcp*[f]{Early stopping check}){$r_t < 10^{-6}$ \textbf{and} $\max |\Delta x_i| < 10^{-6}$}{
                    \textbf{break}
                }
            }
            $\mathcal{G}_{\mathrm{new}} \gets (\mathcal{V}, \mathcal{E}_{\mathrm{new}})$\;
            \If(\tcp*[f]{Check if pruning occurred}){$\mathrm{pruned}$}{
                \Return $\textsc{RecursiveDecayingConfidencePruning}(\mathcal{G}_{\mathrm{new}})$ \tcp*{Recurse on new graph}
            }
        }

        \Return $\mathcal{G}_{\mathrm{new}}$
    }
    \caption{Recursive Neighbor Pruning with Decaying Confidence (RNP-DC)}
    \label{algo:RNP-DC}
\end{algorithm}

To accelerate community decoupling and incorporate a physically realistic model of social homophily, we introduce Recursive Neighbor Pruning with Decaying Confidence (RNP-DC), detailed in Algorithm \ref{algo:RNP-DC}. The baseline RNP evaluates the network topology strictly \emph{after} the dynamics reach a steady state, focusing solely on asymptotic structural polarization. Classical spectral clustering and non-dynamical approaches lack physical information on transient dynamics, focusing only on the graph structure. RNP-DC introduces a state-dependent dynamic evaluating the transient \emph{rate} of local agreement, similar to \cite{morarescuOpinionDynamicsDecaying2011}. It assesses not just whether agents eventually align, but how rapidly they form a local consensus.

The algebraic connectivity of a subgraph dictates the speed of information diffusion. Agents in a dense core influence each other faster than agents separated by a bottleneck. We map this transient behavior to structural affinity using a geometrically shrinking tension threshold, $r_t = R \rho^t$ ($\rho \in (0,1)$). This sequence is the discrete-time equivalent to continuous exponential decay $r(t) = R e^{-\lambda t}$.

Building on the decaying confidence model for community detection \cite{morarescuOpinionDynamicsDecaying2011}, we modify the static interaction in \eqref{eq:dynamic_single_agent} into a state-dependent switching system. The graph topology $\mathcal{G}(t) = (\mathcal{V}, \mathcal{E}(t))$ evolves with the agent states at each discrete step $t \in \mathbb{N}$:
\begin{equation}\label{eq:dynamic_decaying}
    x_i(t+1) = x_i(t) + h \left( \frac{1}{d_i(t)} \sum_{j \in \mathcal{N}_i(t)} \big( s(x_j(t)) - x_i(t) \big) \right),
\end{equation}
where $x_i(t+1) = x_i(t)$ if an agent becomes completely isolated ($d_i(t) = 0$). The time-varying neighborhood $\mathcal{N}_i(t)$ and the instantaneous degree $d_i(t)$ are determined by the dynamically pruned edge set. An initial edge $(i,j) \in \mathcal{E}(0)$ is permanently severed at the first time step $\tau$ when the disagreement between the agents' public signals exceeds the decaying threshold. The logical condition for an edge to exist at step $t$ is:
\begin{equation*}\label{eq:edge_cut_condition}
    (i,j) \in \mathcal{E}(t) \iff \forall \tau \in \{0, 1, \dots, t\}, \ |s(x_i(\tau)) - s(x_j(\tau))| \leq r_\tau.
\end{equation*}

\begin{figure}[t]
    \centering

    \begin{subfigure}[t]{0.66\textwidth}
        \includegraphics{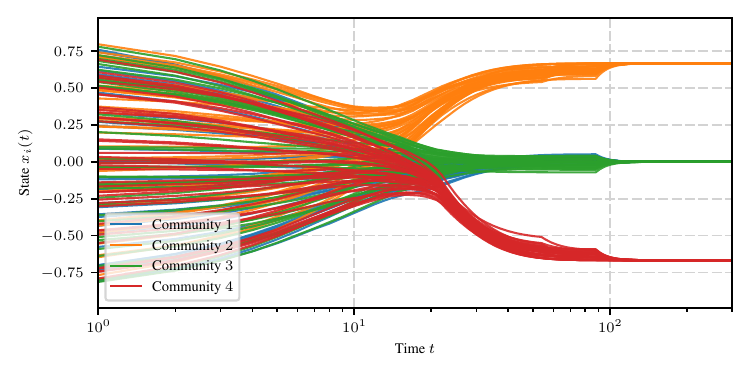}
        \caption{Opinion dynamics of the system with dynamic \eqref{eq:dynamic_decaying}.}
        \label{fig:opinion_dynamics_RNP_DC}
    \end{subfigure}
    \begin{subfigure}[t]{0.33\textwidth}
        \centering
        \begin{tikzpicture}
            \begin{axis}[
                    scale = 1,
                    axis lines=middle,
                    xlabel={$x$},
                    ylabel={$s(x)$},
                    xmin=-1.1, xmax=1.1,
                    ymin=-1.1, ymax=1.1,
                    xtick={-1, 0, 1},
                    ytick={-1, 0, 1},
                    width=\textwidth,
                    height=\textwidth,
                    grid=major,
                    legend pos=north west,
                    declare function={
                            n = 3;
                            delta = 0.01;
                            w = 2.0 / n;
                            idx(\x) = floor((\x + 1.0 - delta) / w);
                            ramp_start(\x) = idx(\x) * w + w - 1.0 - delta;
                            ramp(\x) = min(1, max(0, (\x - ramp_start(\x)) / (2.0 * delta)));
                            smoothstep(\x) = -1.0 + w * (idx(\x) + 0.5 + ramp(\x));
                        },
                ]

                \addplot[domain=-1:1, samples=2, dashed, color=gray] {x};
                \addlegendentry{$y=x$}
                \addplot [
                    domain=-0.99:0.99,
                    samples=1000, 
                    thick,
                    myblue,
                    no marks
                ] {smoothstep(x)};
                \addlegendentry{$s(x)$}

            \end{axis}
            \path (current bounding box.south) ++(0, -1.cm);
        \end{tikzpicture}

        \caption{Signal function $s(x)$ used in the system.}
        \label{fig:signal_function_RNP_DC}
    \end{subfigure}
    \caption{\textbf{States trajectories for the RNP-DC algorithm dynamics with multi-plateau interaction function.} (a) Time evolution of the agents' internal opinions $x_i(t)$ plotted on a logarithmic time scale. As the system evolves, edges are dynamically severed if the disagreement between neighbors exceeds the geometrically shrinking threshold $r_t = R \rho^t$. This decoupling allows the four ground-truth communities to achieve exact internal synchronization across the available invariant plateaus. (b) The multi-plateau signal function $s(x)$ utilized in the update rule, featuring three stable fixed points. The network is a synthetic SBM with $N=160$ nodes (four blocks of $40$), intra-community probability $p_{\text{in}} = 0.80$, and inter-community probability $p_{\text{out}} \le 0.02$. The dynamics were simulated with interaction gain $K=2.5$, step size $h=0.1$, initial tension radius $R=4.0$, and decay rate $\rho=0.98$.}
    \label{fig:rnp_dc_dynamics}
\end{figure}
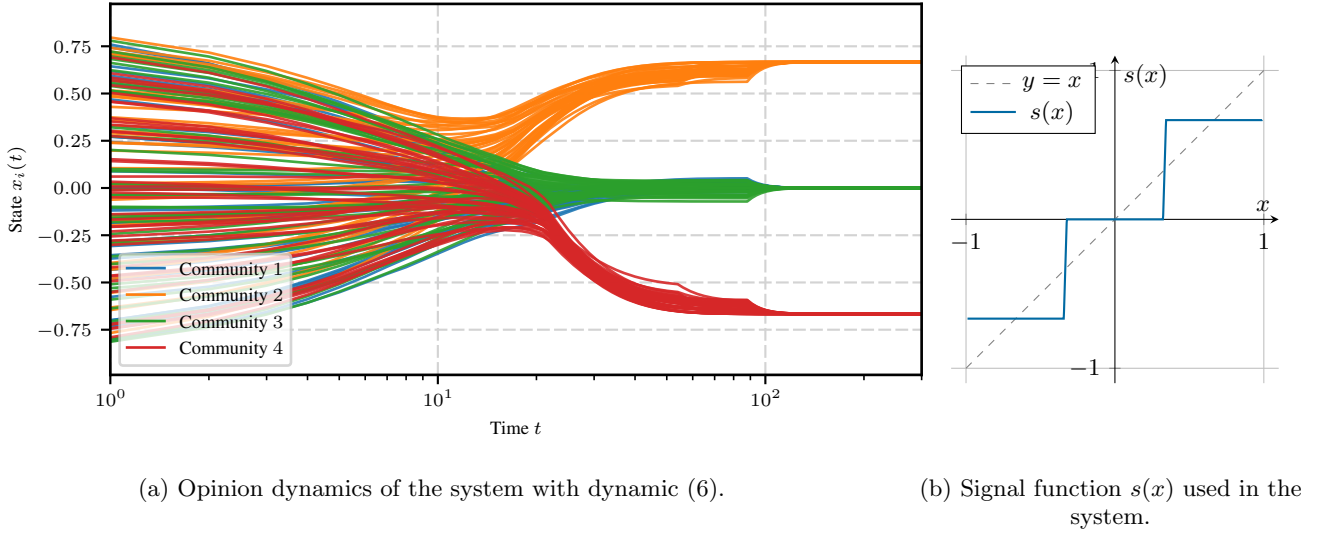
This mechanism operates as a race between state convergence and mathematical threshold decay. For two agents within a tightly knit cluster, the high density of paths drives their states together faster than the boundary $r_t$ collapses, preserving their link. Conversely, for connections spanning sparse inter-community cuts, the flow of information delays local consensus. The geometrically shrinking threshold overtakes their slow convergence, cutting bridging edges well before the system reaches an asymptotic steady state.

Early topological decoupling dynamically isolates the emerging communities. Without contradictory perturbations from neighboring groups, agents in an isolated cluster achieve exact internal synchronization at a stable fixed point of the signal function.

\begin{remark}
    This exact synchronization enables an algorithmic optimization. In baseline RNP, the network remains fully connected during the transient phase, restricting us to binary, two-plateau signal functions (splitting the network above and below zero) to ensure polarization. However, because RNP-DC dynamically isolates the subgroups, we can use a multi-plateau signal function (e.g., a staircase function with multiple stable fixed points spread across $[-1,1]$). By doing so, the network splits into distinct communities simultaneously, with each group exactly synchronizing on a possibly different plateau. Community detection then simplifies to grouping nodes by their final plateau value, reducing the need for deep recursive bisections and accelerating the clustering process. However, as illustrated in Figure~\ref{fig:opinion_dynamics_RNP_DC}, nothing prevents two groups from obtaining the same plateau value, still requiring a recursive bisection.
\end{remark}

Algorithm \ref{algo:RNP-DC} can also be adapted for centralized computation, like Algorithm \ref{algo:RNP_Centralized}. A centralized simulation replaces agent-level loops with vectorized state updates and applies a binary mask to the adjacency matrix $\bm{A}$ at each time step. This eliminates redundant calculations, simulating the dynamics and topological adaptation efficiently for large networks.

\begin{remark}[Stochastic Mixing via Simulated Annealing]
    Because the RNP-DC state update in \eqref{eq:dynamic_decaying} is deterministic, it remains sensitive to initial conditions. If two strongly connected agents initialize with opposite views, the algorithm might permanently sever their tie before the local consensus can reconcile their states. To mitigate this, we can draw inspiration from simulated annealing: introducing a decaying noise term to the state updates injects stochastic fluctuations into the early transient phase. This mixing allows the system to escape accidental early polarizations. As the noise attenuates, the underlying network topology takes over, driving consensus before the decaying threshold $r_t$ freezes the graph structure.
\end{remark}

\subsection{The Statistical Approach: Score-based Edge Reliability (SER)}

\begin{algorithm}[t]
    \SetAlgoLined
    \DontPrintSemicolon
    \SetKwInOut{Input}{Input}\SetKwInOut{Output}{Output}\SetKwInOut{Parameter}{Parameter}
    \SetKwProg{Fn}{Function}{}{end}

    \Input{Graph $\mathcal{G} = (\mathcal{V}, \mathcal{E})$}
    \Output{Reliability map $\bm{P} \in [0,1]^{|\mathcal{E}|}$ for directed edges $(i,j)$}
    \Parameter{Signal func. $s$, Step size $h$, Max steps $T_{\max}$, Number of Iterations $M$, Tension threshold $\varepsilon$, Cut threshold $\tau$}

    \Fn{\textsc{ScoreBasedEdgeReliability}($\mathcal{G}$)}{
        \ForEach{$ \mathrm{directed \ edge} \ (i,j) \in \mathcal{E}$}{
            $C_{i \to j} \gets 0$ \tcp*{Initialize directed trust counter}
        }

        \For(\tcp*[f]{Social learning phase}){$ \mathrm{iteration} \ m = 1$ \KwTo $M$}{
            \lForEach(\tcp*[f]{Initialize new topic opinions}){$\mathrm{agent} \ i \in \mathcal{V}$}{
                $x_i \gets \mathcal{U}(-1, 1)$
            }

            \For(\tcp*[f]{Transient dynamics to steady state}){$t = 1$ \KwTo $T_{\max}$}{
                \lForEach(\tcp*[f]{Broadcast public signal}){$\mathrm{agent} \ i \in \mathcal{V}$}{
                    $s_i \gets s(x_i)$
                }
                \ForEach{$\mathrm{agent} \ i \in \mathcal{V}$}{
                    $\Delta x_i = \frac{1}{d_i} \sum_{k \in \mathcal{N}_i} \left( s_k - x_i \right)$ \tcp*{Calculate local influence}
                    $x_i \gets x_i + h \Delta x_i$ \tcp*{Update continuous internal state}
                }

                \lIf(\tcp*[f]{Early stopping check}){$\max_{i \in \mathcal{V}} |\Delta x_i| < 10^{-6}$}{
                    \textbf{break}
                }
            }

            \ForEach(\tcp*[f]{Evaluate post-discussion structural tension}){$ \mathrm{directed \ edge} \ (i, j) \in \mathcal{E}$}{
                $\Delta_{i \to j} \gets |x_i - s_j| \ (\mathrm{or \ alternatively} \ |s_i - s_j|)$ \tcp*{Evaluate tension (Opinion vs. Action)}
                \lIf(\tcp*[f]{Increment directed agreement counter}){$\Delta_{i \to j} < \varepsilon$}{
                    $C_{i \to j} \gets C_{i \to j} + 1$
                }
            }
        }

        $\mathcal{E}_{\mathrm{new}} \gets \mathcal{E}$\;
        \ForEach(\tcp*[f]{Statistical filtering}){$\mathrm{directed \ edge} \ (i, j) \in \mathcal{E}$}{
            $P_{i \to j} \gets C_{i \to j} / M$ \tcp*{Compute reliability score}

            \If(\tcp*[f]{Sever consistently conflicting ties}){$P_{i \to j} < \tau$}{
                $\mathcal{E}_{\mathrm{new}} \gets \mathcal{E}_{\mathrm{new}} \setminus \{(i,j)\}$
            }
        }
        \Return $\mathcal{G}_{\mathrm{new}} = (\mathcal{V}, \mathcal{E}_{\mathrm{new}})$, reliability scores $\bm{P}$
    }
    \caption{Score-based Edge Reliability (SER)}
    \label{algo:SER}
\end{algorithm}

While RNP and RNP-DC use nonlinear dynamics to fracture the network, relying on recursive bisections carries the risk of premature hard cuts. Because the discrete-time dynamics is highly sensitive to initial conditions, a single realization samples only a narrow trajectory within the broader attractor landscape. If two strongly connected agents within the same core community initialize with opposing extreme views, the transient dynamics may temporarily drive them further apart. The RNP-DC protocol penalizes this transient disagreement by permanently severing the edge, degrading the structural integrity of the core community for all subsequent integration steps.

To systematically resolve this stochastic sensitivity, we shift our analytical perspective: rather than asking ``Which discrete group does this node strictly belong to after a single run?'', we ask ``How structurally resilient is this specific edge across a diverse distribution of initial conditions?''

Sociologically, this maps to the observation that individuals rarely sever strong social ties over a single disagreement. A robust connection withstands occasional conflict stemming from divergent initial beliefs. Interpersonal links are severed only when agents consistently fail to achieve consensus across multiple independent issues.

Algorithm~\ref{algo:SER} formulates this logic as the Score-based Edge Reliability (SER) approach. Instead of dynamically fracturing the graph during the transient phase, SER maintains a static topology $\mathcal{G}$ and executes $M$ independent trials of the social learning dynamics. Each realization $m \in \{1, \dots, M\}$ is initialized with a new uniform random state $\bm{x} \sim \mathcal{U}(-1, 1)^N$, analogous to the network debating $M$ distinct, uncorrelated topics. Following the convergence of each iteration, agents evaluate the local structural tension along their incident edges.

\subsubsection*{Evaluation metrics: Action vs. Action and Opinion vs. Action}

Agents can evaluate post-discussion tension in two ways:

\begin{figure}[t]
    \centering
    \begin{subfigure}[t]{0.42\textwidth}
        \includegraphics[width=\textwidth]{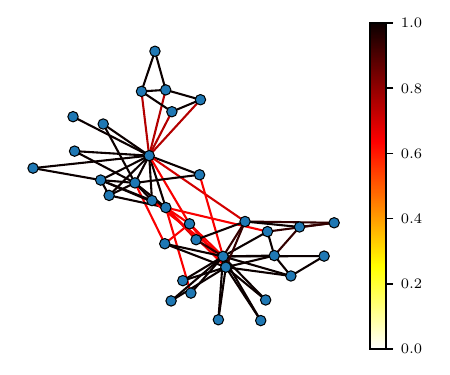}
        \caption{Symmetric agreement matrix based on Action vs. Action ($|s_i - s_j|$).}
    \end{subfigure}
    \hfil
    \begin{subfigure}[t]{0.35\textwidth}
        \includegraphics[width=\textwidth]{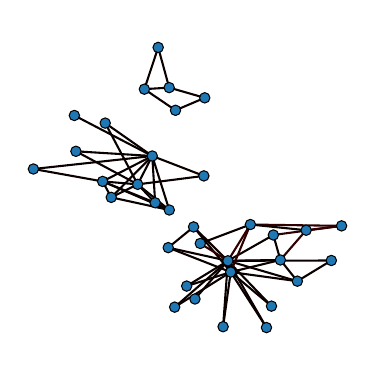}
        \caption{Obtained undirected graph with cut threshold $\tau = 0.9$.}
    \end{subfigure}
    \caption{\textbf{Structural tension and partitions via Action vs. Action evaluation.} The resulting graph of SER algorithm on Zachary's Karate club network \cite{zacharyInformationFlowModel1977a}. Boundary nodes are often fully absorbed into the core communities due to their saturated public signals.}
    \label{fig:SER_action}
\end{figure}

\textbf{1. Symmetric Alignment (Action vs. Action):} Comparing the observable signals of adjacent agents yields the tension $\Delta_{ij} = |s_i - s_j|$. If $\Delta_{ij} < \varepsilon$, the agents pulled each other into the same invariant hypercube. Over $M$ independent realizations, this yields a symmetric probability score $P_{ij} = P_{ji} \in [0,1]$. Figure~\ref{fig:SER_action} shows this filters out weak boundary edges. Edges inside a community synchronize ($P_{ij} \approx 1$), while inter-community edges fracture ($P_{ij} \approx 0.5$). Frontier nodes are heavily influenced by dense neighbors, so they often saturate their public action ($s \approx \pm 1$). This approach classifies frontier nodes into the adjacent community with the higher degree.

\begin{figure}[t]
    \centering
    \begin{subfigure}[t]{0.42\textwidth}
        \includegraphics[width=\textwidth]{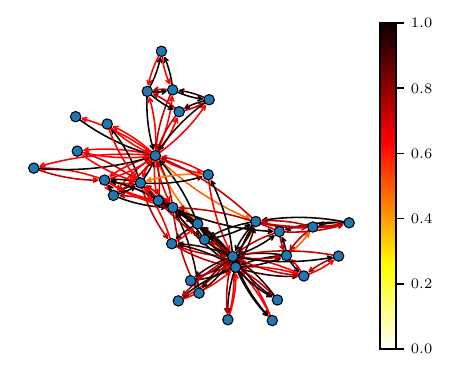}
        \caption{Asymmetric agreement matrix based on Opinion vs. Action ($|x_i - s_j|$).}
    \end{subfigure}
    \hfil
    \begin{subfigure}[t]{0.35\textwidth}
        \includegraphics[width=\textwidth]{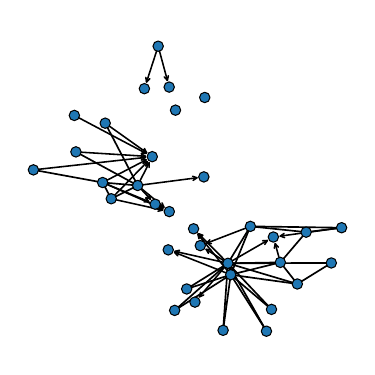}
        \caption{Obtained directed graph with cut threshold $\tau = 0.9$.}
    \end{subfigure}
    \caption{\textbf{Isolating dense cores and structural ambiguity via Opinion vs. Action evaluation.} The evaluation generates a directed trust graph on Zachary's Karate club network \cite{zacharyInformationFlowModel1977a}. Core nodes remain highly interconnected (Strongly Connected Components), while moderate frontier nodes dynamically reject the extremist cores, naturally identifying the exact community boundaries.}
    \label{fig:SER_opinion}
\end{figure}

\textbf{2. Asymmetric Alignment (Opinion vs. Action):} To achieve a higher-resolution map of the community boundaries, an agent evaluates its \emph{private internal opinion} against its neighbor's \emph{public action}:
\begin{equation*}
    \Delta_{i \to j} = |x_i - s_j|.
\end{equation*}
Unlike the previous metric, this evaluation is asymmetric ($\Delta_{i \to j} \neq \Delta_{j \to i}$), transforming the undirected interaction graph into a \emph{directed} reliability map, as illustrated in Fig.~\ref{fig:SER_opinion}.

This asymmetry captures the sociological dynamic between extremist members and moderate frontier nodes. For example, consider a highly connected core node $i$ ($x_i \approx 1, s_i = 1$) connected to a structurally ambiguous frontier node $j$. Under the pressure of the core, the frontier node projects a saturated public agreement ($s_j = 1$), but due to conflicting pull from other communities, its internal opinion remains moderate and uncommitted ($x_j \approx 0.3$).

For the core node, the tension is low: $\Delta_{i \to j} = |1 - 1| = 0$. Observing only the public action, the core node trusts the frontier node. For the frontier node, the tension is high: $\Delta_{j \to i} = |0.3 - 1| = 0.7$. The frontier node recognizes its internal belief clashes with the core's public stance.

Thresholding this directed reliability map ($P_{i \to j} < \tau$) causes moderates to sever incoming ties from extremists, while outward ties from the core to the boundary remain. Figure~\ref{fig:SER_opinion}(b) shows the core communities emerge as isolated strongly connected components. Boundary nodes move to peripheral ``sink'' positions, preventing them from inflating the core communities.

\subsection{Methodological Comparison}

\begin{table*}[t]
    \caption{\label{tab:algo_comparison} Methodological comparison of social learning community detection algorithms.}
    \begin{ruledtabular}
        \footnotesize{
            \begin{tabular}{p{0.14\linewidth} p{0.27\linewidth} p{0.27\linewidth} p{0.27\linewidth}}
                \textbf{Name}                                    & \textbf{Recursive Neighbor Pruning (RNP)}                                                                     & \textbf{RNP with Decaying Confidence (RNP-DC)}                                                                                             & \textbf{Score-based Edge Reliability (SER)}                                                                                    \\
                \hline
                \textbf{Focus}                                   & Asymptotic agreement.                                                                                         & Transient speed of persuasion.                                                                                                             & Resilience across multiple topics.                                                                                             \\
                \textbf{Graph Topology}                          & Piecewise-static (pruned only upon reaching equilibrium).                                                     & State-dependent (pruned on the fly during the transient phase).                                                                            & Strictly static throughout all realizations.                                                                                   \\
                \textbf{Partition Logic}                         & Node-centric hard cuts based on disjoint invariant hypercubes.                                                & Dynamic isolation based on shrinking tension threshold ($r(t)$).                                                                           & Edge-centric probabilistic filtering ($P_{ij} \geq \tau$).                                                                     \\
                \hline
                \textbf{Modeling} \newline \textbf{Advantages}   & Nonlinear extension of standard spectral bisection. Random initialization explores diverse structural basins. & Early decoupling protects dense cores from external noise. Achieves \emph{exact synchronization}, enabling multi-plateau signal functions. & Eliminates error from hard cuts. Naturally exposes frontier nodes and overlapping mesoscale structures.                        \\
                \textbf{Modeling} \newline \textbf{Limitations}  & Highly vulnerable to errors from premature hard cuts. Arbitrarily divides boundary nodes.                     & Introduces non-topological decay parameters ($R, \rho$) that require tuning relative to the spectral gap of the graph.                     & Requires a secondary clustering step (e.g., extracting connected components from the filtered graph) to yield discrete groups. \\
                \hline
                \textbf{Computational} \newline \textbf{Profile} & Conceptually simple to implement. Requires multiple iterations to confirm indivisibility.                     & Dynamic edge pruning progressively accelerates numerical integration, but per-step edge evaluations limit scalability.                     & High aggregate computational cost, but trivially parallelizable.                                                               \\
            \end{tabular}
        }
    \end{ruledtabular}
\end{table*}

All three algorithms use nonlinear symmetry breaking, but their partition logics yield different trade-offs. Table \ref{tab:algo_comparison} summarizes their modeling advantages, limitations, and computational profiles.

The selection of the appropriate algorithm depends on the constraints of the target network. For applications requiring strict analytical reliability and the nuanced identification of ambiguous boundary nodes, SER provides the most robust mapping of the topological landscape. Conversely, for applications prioritizing computational speed, or in dynamic environments where rapid topological decoupling is advantageous, RNP-DC is highly effective. The subsequent section validates these theoretical properties numerically and benchmarks their clustering performance against standard modularity metrics.

\section{Numerical Section}\label{sec:numerical_validation}
While centralized heuristics like the Louvain method achieve near-linear time complexity, these simulations demonstrate that localized, privacy-preserving interactions can recover macroscopic community structures with equivalent accuracy. The objective is to achieve decentralized coordination rather than to surpass the global computational speed of centralized algorithms.

\subsection{Benchmark Setup and the Choice of ABCD Graphs}
\label{sec:benchmark_setup}

To systematically evaluate the algorithms, we require synthetic networks with a known ground-truth community structure. Historically, the LFR (Lancichinetti-Fortunato-Radicchi) benchmark \cite{lancichinettiBenchmarkGraphsTesting2008} has been the standard for this task. The LFR model generates graphs based on power-law distributions for both node degrees and community sizes, controlled by a mixing parameter $\mu$ that dictates the proportion of a node's edges that connect outside its community. This global parameter $\mu$ is similar to the parameters $\mu_i$ defined in Proposition~\ref{prop:iss_cell_dev}.

However, from the perspective of our discrete-time dynamics \eqref{eq:dynamic_vector}, a community is only robust if it acts as a cohesive local consensus block. According to Theorem \ref{thm:sharp_threshold}, a graph or subgraph can fracture if the interaction gain satisfies $K > 1/\lambda_{N-1}$, where $\lambda_{N-1}$ is the second-largest eigenvalue of its random-walk matrix $\bm{P}$.

For a community detection algorithm based on nonlinear symmetry breaking to succeed, a strict spectral hierarchy must exist: the global critical gain required to split the entire network ($K^*_{\mathrm{global}} = 1/\lambda_{N-1}(G)$) must be strictly smaller than the critical gain required to fracture the internal ground-truth communities ($K^*_{\mathrm{sub}} = 1/\lambda_{N-1}(G_c)$). If $K^*_{\mathrm{global}} > K^*_{\mathrm{sub}}$, tuning $K$ to break the global consensus may simultaneously—or prematurely—shatter the ground-truth communities, leading to unavoidable over-segmentation.

\begin{figure}[t]
    \centering
    \begin{subfigure}[t]{0.48\textwidth}
        \includegraphics[width=\textwidth]{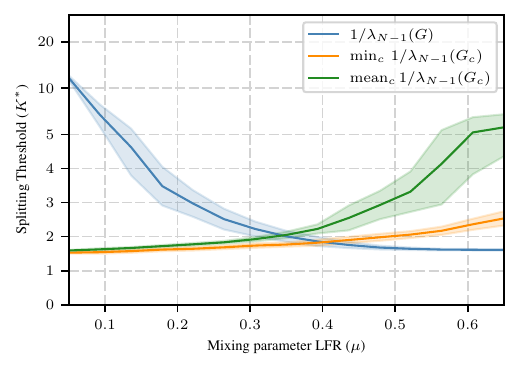}
        \caption{LFR Benchmark: $\beta = 1.5, \gamma = 3$}
        \label{fig:lfr_threshold}
    \end{subfigure}
    \hfill
    \begin{subfigure}[t]{0.48\textwidth}
        \includegraphics[width=\textwidth]{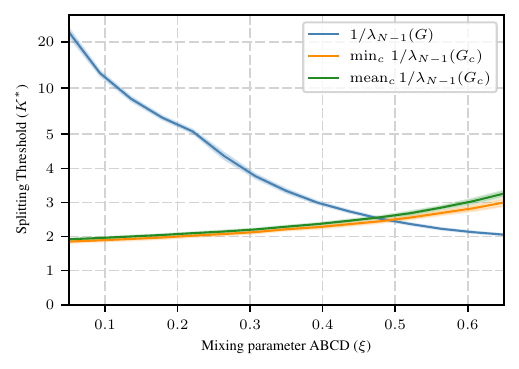}
        \caption{ABCD Benchmark: $\beta = 1.5, \gamma = 3$}
        \label{fig:abcd_threshold}
    \end{subfigure}
    \caption{\textbf{Spectral difference of ground truth communities between the LFR vs. ABCD benchmarks.} The blue line represents the global critical splitting gain $K^*_{\mathrm{global}} = 1/\lambda_{N-1}(G)$, while the orange and green lines represent the minimum and mean critical gains required to fracture the isolated ground-truth communities. (a) In LFR graphs, the curves cross prematurely at $\mu \approx 0.35$, meaning ground-truth communities lose internal structural integrity before the global network does. (b) In ABCD graphs, the strict spectral hierarchy ($K^*_{\mathrm{global}} < K^*_{\mathrm{sub}}$) is preserved almost up to the theoretical detectability limit of $\xi = 0.5$.}
    \label{fig:lfr_vs_abcd}
\end{figure}

As illustrated in Figure \ref{fig:lfr_threshold}, evaluating the spectrum of LFR graphs reveals this necessary hierarchy fails early. The global threshold curve $K^*_{\mathrm{global}}$ (blue) crosses the mean sub-community threshold (green) near $\mu \approx 0.35$. Past this point, the induced subgraphs of LFR ground-truth communities become poorly connected. The nonlinear dynamics artificially segment these weak communities before identifying the true graph partition.

To resolve this, we utilize the Artificial Benchmark for Community Detection (ABCD) \cite{kaminskiArtificialBenchmarkCommunity2021}. The ABCD model preserves the realistic power-law properties of LFR but replaces $\mu$ with a mixing parameter $\xi \in (0,1)$ and enforces well-connected induced subgraphs. As shown in Figure \ref{fig:lfr_vs_abcd}(b), ABCD preserves the spectral hierarchy $K^*_{\mathrm{global}} < K^*_{\mathrm{sub}}$, making it a mathematically sound benchmark for dynamic and spectral clustering.

Analyzing the crossing point in the ABCD model empirically validates our theoretical bounds on community robustness. The mixing parameter $\xi$ represents the expected ratio of external to total degree ($d_i^E / d_i$). From Proposition \ref{prop:cell_deviation_bound}, permanent confinement in an invariant hypercube requires the internal restoring force to dominate the external perturbation. If $\xi > 0.5$, nodes possess more external connections than internal ones ($d_i^E > d_i^I$). Under these conditions, the aggregate pull of the external network overpowers the internal consensus for any reasonably smooth signal function.

This is reflected in Figure \ref{fig:abcd_threshold}, where the threshold curves intersect near $\xi = 0.5$. This point represents a hard physical and topological bottleneck. Beyond this detectability limit, the local social learning dynamics is theoretically guaranteed to fail at preserving strict community boundaries. Consequently, our performance evaluations focus strictly on the feasible structural regime where $\xi \leq 0.5$.

In the following experiments, we generate ABCD graphs with $N=1000$ nodes, a degree power-law exponent $\gamma=2.5$, a community size power-law exponent $\beta=1.5$, a minimum degree of $10$, and community sizes constrained between $20$ and the maximum possible.

\subsection{Comparative Performance}
\label{sec:comparative_performance}

\begin{table}[t]
    \centering
    \caption{Simulation Parameters and Algorithm Hyperparameters}
    \label{tab:sim_params}
    \begin{tabular}{@{}lp{0.65\columnwidth}@{}}
        \toprule
        \textbf{Category}        & \textbf{Parameter Values}                                                                                                                             \\
        \midrule
        \textbf{ABCD Benchmark}  & Nodes $N=1000$, degree exp. $\gamma=2.5$, comm. size exp. $\beta=1.5$, min. degree $= 10$, min. comm. size $= 20$, independent trials per $\xi = 20$. \\ \addlinespace
        \textbf{Global Dynamics} & Signal function $s(x) = \min(1, \max(-1, Kx)$, Interaction gain $K=4$, step size $h=0.5$, time steps $T=300$.                                         \\ \addlinespace
        \textbf{RNP}             & Number of retrials $N_{\text{retry}} = 50$.                                                                                                           \\ \addlinespace
        \textbf{RNP-DC}          & Initial Radius $R=20$, Decay rate $\rho=0.99$, Number of retrials $N_{\text{retry}} = 50$.                                                            \\ \addlinespace
        \textbf{SER}             & Number of Iterations $M= 200$, Cut threshold $\tau=1$, Tension threshold $\epsilon=0.5$.                                                              \\
        \bottomrule
    \end{tabular}
\end{table}

To evaluate the limits of the proposed decentralized frameworks, we analyze their performance across a spectrum of topological difficulties. Using the ABCD benchmark as justified in Section~ \ref{sec:benchmark_setup}, we vary the mixing parameter $\xi$ from $0.01$ (highly cohesive, isolated communities) to $0.60$ (networks heavily dominated by inter-community noise). To ensure strict reproducibility, the specific structural properties of the generated graphs (e.g., degree distribution $\gamma=2.5$, community size distribution $\beta=1.5$) and all algorithm hyperparameters are summarized in Table~\ref{tab:sim_params}. We compare our three dynamic approaches—RNP, RNP-DC, and SER—against four established baselines: Louvain \cite{blondelFastUnfoldingCommunities2008} (modularity maximization), Leiden \cite{traagLouvainLeidenGuaranteeing2019} (an optimized refinement of Louvain), standard Label Propagation \cite{raghavanLinearTimeAlgorithm2007} (a decentralized heuristic), and exact Spectral Clustering \cite{pothenPartitioningSparseMatrices1990}.

For the social learning algorithms, agent interactions follow the saturated piecewise-linear signal function $s(x) = \min(1, \max(-1, 4x))$. This yields a local interaction gain of $K=4$ at the unstable origin, driving discrete-time updates with a step size of $h=0.5$. We assess the resulting partitions using Normalized Mutual Information (NMI), Relative Segmentation Error ($(M_{\mathrm{pred}} - M_{\mathrm{truth}}) / M_{\mathrm{truth}}$), Modularity ($Q$), and Average Conductance.

\begin{figure*}[t]
    \centering
    \begin{subfigure}[t]{0.48\textwidth}
        \includegraphics{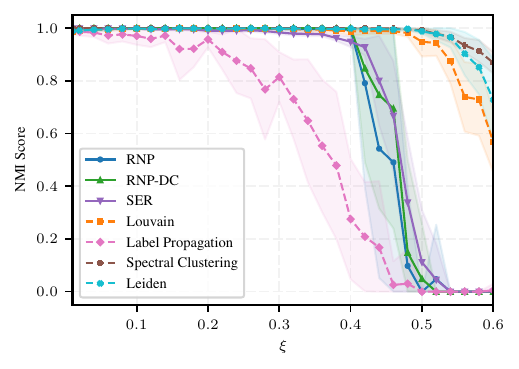}
        \caption{Normalized Mutual Information (NMI)}
        \label{fig:bench_nmi}
    \end{subfigure}
    \begin{subfigure}[t]{0.48\textwidth}
        \includegraphics{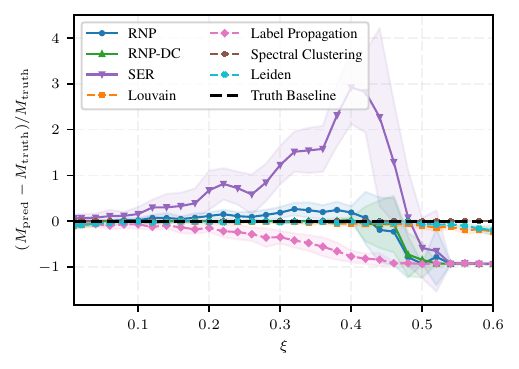}
        \caption{Relative Segmentation Error}
        \label{fig:bench_seg_error}
    \end{subfigure}
    \begin{subfigure}[t]{0.48\textwidth}
        \includegraphics{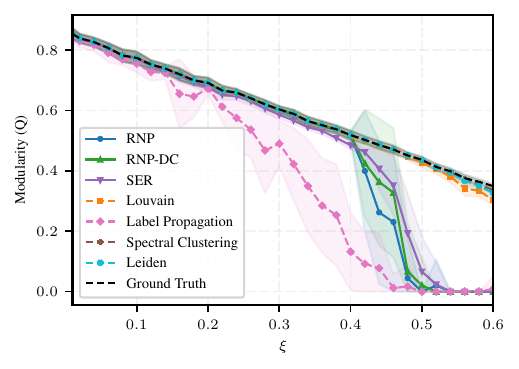}
        \caption{Modularity ($Q$)}
        \label{fig:bench_mod}
    \end{subfigure}
    \begin{subfigure}[t]{0.48\textwidth}
        \includegraphics{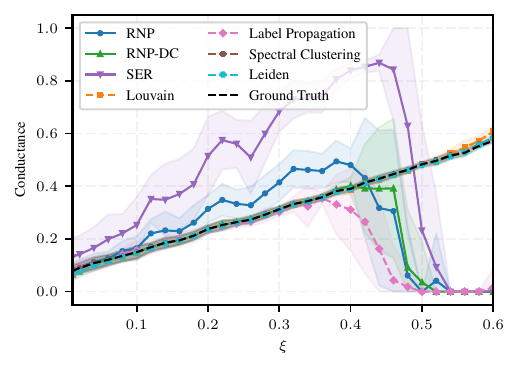}
        \caption{Average Conductance}
        \label{fig:bench_cond}
    \end{subfigure}
    \caption{\textbf{Benchmarking decentralized algorithms against literature baselines.} Performance evaluation on ABCD graphs ($N=1000$) as a function of the mixing parameter $\xi$. (a) Normalized Mutual Information (NMI). (b) Relative Segmentation Error, illustrating the capacity of algorithms to determine the correct number of communities ($M_{\mathrm{Truth}}$). (c) Modularity ($Q$) of the recovered partitions. (d) Average Conductance. The proposed methods (RNP, RNP-DC, SER) perform exactly on par with globally optimized centralized heuristics (Louvain, Leiden, Spectral Clustering) up to the theoretical detectability limit of the ABCD model ($\xi \approx 0.45$) identified in the Figure~\ref{fig:lfr_vs_abcd}. Furthermore, they outperform classical decentralized approaches like Label Propagation.}
    \label{fig:abcd_benchmark}
\end{figure*}

We plot a \emph{Ground Truth} baseline for relative segmentation error, modularity, and average conductance (Figures~\ref{fig:bench_seg_error}, \ref{fig:bench_mod}, and \ref{fig:bench_cond}). We compute these curves directly on the known communities from the ABCD graph generation. This reference tracks how the structural properties of the initial partitions degrade as the inter-community noise parameter $\xi$ increases. For the relative segmentation error (Figure~\ref{fig:bench_seg_error}), this ground truth simply provides a zero-baseline at the exact number of planted communities.

\paragraph{Symmetry Breaking and the Detectability Limit:}
As illustrated in Figure \ref{fig:abcd_benchmark}, the numerical results reveal a clear phase transition that illustrates the theoretical spectral bounds from Section \ref{sec:structural_splitting}. In the structurally discernible regime ($\xi \le 0.40$), local nonlinear interactions destabilize the neutral origin, forcing the network to fracture along its true topological bottlenecks. The proposed algorithms achieve an NMI close to $1.0$. Our strictly local privacy-preserving agents reconstruct ground-truth partitions with an accuracy matching centralized heuristics like Louvain and Leiden. Standard discrete Label Propagation collapses prematurely near $\xi \approx 0.15$.

As the mixing parameter approaches the physical detectability limit ($\xi \approx 0.45$), the dense internal cores begin to dissolve into the external noise. Recalling Proposition \ref{prop:cell_deviation_bound}, the restoring force of the local consensus is ultimately overwhelmed by the conflicting pull of external connections, causing the invariant hypercubes to break down. Without strict state confinement, the decentralized dynamic algorithms lose the structural resolution required to separate true boundary cuts from stochastic topological noise.

\paragraph{Clustering Without a priori information on the topology:}
An advantage of this approach is its ability to endogenously discover the network's macroscopic scale. Standard spectral clustering requires the true number of communities $M_{\mathrm{truth}}$ as an \textit{a priori} input to define the dimension of the projection space. Our agents operate with no global structural priors. Driven by local symmetry breaking and state confinement, RNP and RNP-DC dynamically settle into the correct number of polarized basins, maintaining a relative segmentation error close to zero up to the critical threshold.

For the SER algorithm, we observe a transient phase of over-segmentation (positive relative error) between $\xi=0.30$ and $\xi=0.45$. Because SER statistically evaluates the structural resilience of edges across multiple independent iterations, the increasing topological ambiguity in this regime causes noisy boundary ties to frequently fail the strict trust threshold. Sociologically, this means that rather than forcing arbitrary, erroneous macroscopic merges, SER safely fractures highly contested frontier regions into smaller, isolated subgroups.

As $\xi$ exceeds the critical detectability limit, the relative segmentation error for all three algorithms converges to $-1$. This boundary manifests Theorem \ref{thm:sharp_threshold}. At this point, severe topological mixing ($\lambda_{N-1}$) completely suppresses the local interaction gain ($K=4$). Structural splitting becomes mathematically impossible, and the entire network is pulled into a single globally attractive consensus state ($M_{\mathrm{pred}} = 1$), yielding a relative error of $(1 - M_{\mathrm{truth}})/M_{\mathrm{truth}} \approx -1$.

\paragraph{Structural Optimality:}
Evaluations of Modularity ($Q$) and Conductance confirm that partitions generated by the dynamics \eqref{eq:dynamic_vector} align with optimal structural cuts. Within the feasible structural regime ($\xi \le 0.40$), the purely localized decisions made by RNP, RNP-DC, and SER track the modularity maximized by Louvain, ensuring dense internal cohesion and sparse external boundaries. The artificial drop in Conductance to zero past $\xi = 0.50$ confirms the global consensus takeover: unable to sustain localized disagreements, the network evaluates itself as a single undivided component devoid of external boundaries.

\subsection{Parameter Sensitivity and Local Tuning}
\label{sec:parameter_sensitivity}

\begin{figure}[t]
    \centering
    \begin{subfigure}[t]{0.48\textwidth}
        \includegraphics[width=\textwidth]{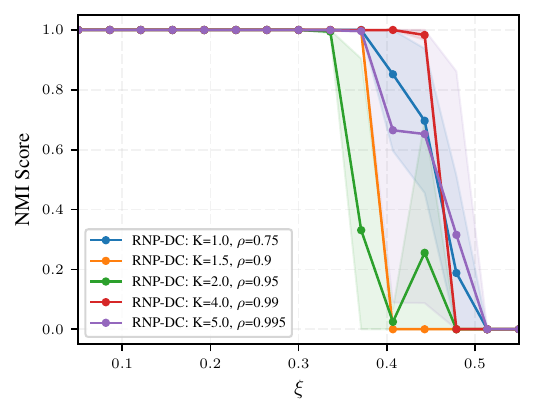}
        \caption{RNP-DC: Normalized Mutual Information (NMI)}
    \end{subfigure}
    \begin{subfigure}[t]{0.48\textwidth}
        \includegraphics[width=\textwidth]{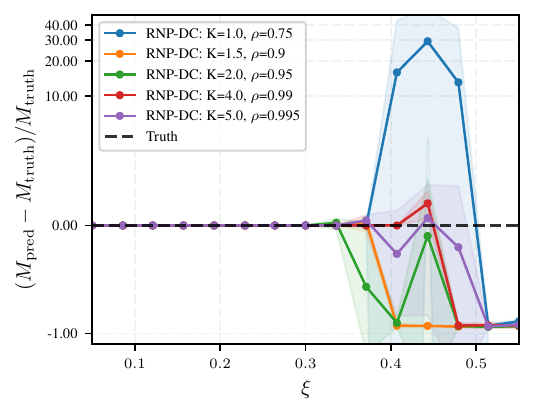}
        \caption{RNP-DC: Relative Segmentation Error}
    \end{subfigure}
    \caption{\textbf{Parameter variation in the RNP-DC algorithm: balancing interaction gain ($K$) and confidence decay ($\rho$).} Evaluating the sensitivity of the algorithm to its main parameters on ABCD graphs. (a) Normalized Mutual Information (NMI). (b) Relative Segmentation Error. An aggressive decay paired with weak interaction ($K=1.0, \rho=0.75$) (this is a particular case of \cite{morarescuOpinionDynamicsDecaying2011}) arbitrarily shatters the network as noise increases, leading to severe over-segmentation. Conversely, a robust interaction gain paired with a slower, deliberate decay ($K=4.0, \rho=0.99$) correctly identifies the exact number of communities up to $\xi=0.45$, before safely yielding to a unified global consensus (-1 relative error) when the physical detectability limit is exceeded.}
    \label{fig:rnpdc_parameter_sensitivity}
\end{figure}

In centralized spectral graph theory, algorithmic performance is frequently optimized by leveraging a global view of the network. If a centralized operator possessed full knowledge of the graph's topology, specifically the algebraic connectivity $1 - \lambda_{N-1}$ and the expected mixing parameter $\xi$, they could analytically calibrate the dynamical parameters of our models. For instance, in the RNP-DC algorithm, the topological discovery relies entirely on a delicate race condition between the physical state convergence (driven by the gain $K$) and the mathematical edge pruning (driven by the decay rate $\rho$ as in \cite{morarescuOpinionDynamicsDecaying2011} where the $K$ is fixed to $K=1$).

Figure \ref{fig:rnpdc_parameter_sensitivity} exposes the highly sensitive nature of this race condition. When the interaction gain is low and threshold decay is overly aggressive (e.g., $K=1.0, \rho=0.75$), the algorithm severs bridging edges long before local subgraphs form a cohesive internal consensus. As structural noise increases ($\xi > 0.4$), this miscalibration causes the algorithm to shatter the graph into meaningless fragments, indicated by the massive spike in relative segmentation error (reaching up to $30 \times M_{\mathrm{truth}}$).

Conversely, moderate parameter pairings (e.g., $K=1.5, \rho=0.9$ or $K=2.0, \rho=0.95$) fail in the opposite direction. As the network becomes more interconnected, the moderate gain is insufficient to overcome topological averaging, and the decay is not slow enough to allow dense cores to isolate themselves. These configurations lose their structural resolution prematurely, plunging to an NMI of $0$ and a relative error of $-1$ (indicating a failure to split) well before the theoretical detectability limit.

The empirical optimum for this network size and density is achieved by combining a strong local interaction with a slow, deliberate decay ($K=4.0, \rho=0.98$). This configuration holds perfect NMI and zero segmentation error up to $\xi=0.45$. When the theoretical threshold is breached, this robust parameter set fails safely: rather than returning a noisy and incorrect partition, local restoring forces are smoothly overcome by the external topology, resulting in a clean transition to a single global consensus.

This sensitivity analysis highlights the challenge of decentralized community detection. In our targeted applications privacy-preserving human social networks computing $\lambda_{N-1}$ or estimating $\xi$ is strictly impossible. Local agents operate blindly, possessing only a self-centered view of their immediate neighborhood.

Relying on hyper-tuned parameters tailored to a specific $\xi$ is therefore not a viable strategy for true decentralized deployment. A practical distributed algorithm must trade the absolute peak performance achievable on a known graph for broad resilience across an unknown spectrum of topologies. This constraint justifies the use of a saturated, robust default gain ($K=4$) to guarantee symmetry breaking. It also validates the necessity of the SER approach (Section \ref{sec:algorithms}). Rather than trusting a single transient trajectory governed by hypothesized temporal parameters ($h$ and $\rho$), SER utilizes statistical repetition to bypass the lack of global information, extracting the true topological boundaries directly from the agents' interactions.

\subsection{Real-World Validation: The Fission of the Ngogo Chimpanzee Network}
\label{sec:real_world}

Before extending our evaluation to empirical datasets, we note that the classical Zachary's Karate Club network \cite{zacharyInformationFlowModel1977a} has already been utilized in Section \ref{sec:algorithms} (see Figs. \ref{fig:SER_action} and \ref{fig:SER_opinion}) to illustrate the foundational mechanics of the SER algorithm. Having established the ability of the algorithm to identify the true structural partition and properly isolate frontier nodes in such foundational benchmarks, we now deploy the framework on two larger, highly nuanced real-world topologies, beginning with a biological animal social network.

We first analyze the social proximity network of the Ngogo chimpanzee community \cite{sandelLethalConflictGroup2026,sandelDataLethalConflict2026}. Recently documented as a rare instance of permanent group fission, this dataset provides a unique look at a social network undergoing a macroscopic structural division.

From a multi-agent systems perspective, this biological phenomenon fits the decentralized clustering objective formulated in Section \ref{sec:problem_formulation}. In this natural system, individuals do not possess a global view of the interaction topology, nor do they communicate explicit group identifiers. Group affiliations are organically inferred and updated through localized, pairwise interactions (e.g., grooming and spatial proximity). The observation by Sandel et al. that the network's division was driven primarily by shifting relational dynamics---rather than exogenous cultural markers---provides a direct natural parallel to the localized symmetry-breaking mechanism captured by our opinion dynamics model \eqref{eq:dynamic_vector}.

\begin{figure}[t]
    \centering
    \begin{subfigure}[t]{0.40\textwidth}
        \includegraphics{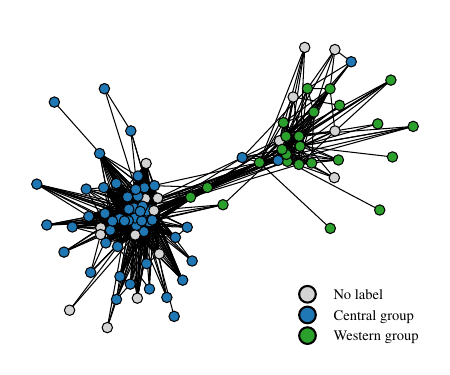}
        \caption{Ngogo Chimpanzee Network (2015) with 2018 True Factions}
        \label{fig:chimp_gt}
    \end{subfigure}
    \hfil
    \begin{subfigure}[t]{0.50\textwidth}
        \includegraphics{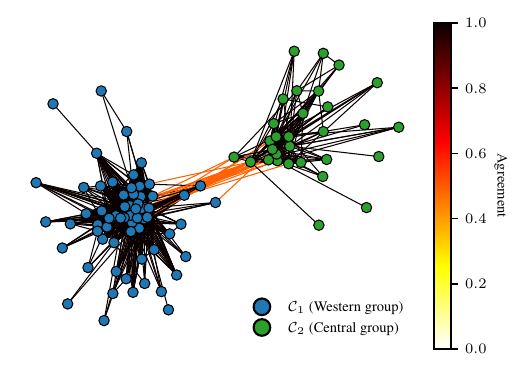}
        \caption{SER Decentralized Edge Evaluation}
        \label{fig:chimp_ser}
    \end{subfigure}
    \caption{\textbf{Comparison of biological ground truth and SER output on the Ngogo Chimpanzee Network.} (a) The observed proximity network in 2015, the transition year of the structural bifurcation, annotated with the eventual ground-truth affiliations (Central vs. Western factions of 2018 splitting \cite{sandelLethalConflictGroup2026}). (b) The emergent state evaluated by the Score-based Edge Reliability (SER) algorithm ($K=4, h=0.5, \tau=1$). The algorithm successfully recovers the two primary communities strictly via local interactions. The colormap on the edges represents the probabilistic agreement matrix ($P_{ij}$), indicating lower structural reliability (orange) on the bridging edges that ultimately fractured.}
    \label{fig:chimp_validation}
\end{figure}

For our evaluation, we examine the proximity network from 2015, identified by the original study as the transition year when the group began to exhibit structural polarization. Figure \ref{fig:chimp_gt} illustrates this unweighted interaction topology, with nodes colored according to their eventual ground-truth affiliations (the Central and Western factions that ultimately break apart in 2018 \cite{sandelLethalConflictGroup2026}).

Applying the SER algorithm to this static snapshot, agents were initialized with random continuous states across $100$ independent realizations to simulate the inherent variability of local interactions. The system was parameterized with a moderate interaction gain ($K=4$), a step size of $h=0.5$, and an edge-cutting threshold of $\tau=1$.

As depicted in Figure \ref{fig:chimp_ser}, the SER protocol recovers the macroscopic partition of the network purely through local consensus dynamics. Driven by the underlying topological bottlenecks, the nonlinear interaction destabilizes the neutral origin, guiding structurally cohesive subgroups into opposing invariant hypercubes and partitioning the network without requiring global eigenvalue computations.

The probabilistic edge filtering intrinsic to the SER approach yields much more information than standard recursive bisections. The colormap in Figure \ref{fig:chimp_ser} projects the continuous agreement matrix ($P_{ij}$) onto the network edges. Links embedded within the dense Central and Western cores exhibit near-perfect agreement ($P_{ij} \approx 1$, dark red), demonstrating robust structural resilience against the stochastic initialization of internal states. In contrast, the bridging edges spanning the two emerging factions register significantly lower agreement scores ($P_{ij} < 0.6$, orange). Because these boundary nodes are subjected to competing influences from opposing dense cores, their transient trajectories exhibit high variability. The SER algorithm naturally highlights these specific topological links as the weakest structural points.

This empirical application demonstrates a practical advantage of the dynamical approach. While standard centralized spectral clustering typically requires the target number of communities $k$ as an \textit{a priori} input to force a rigid discrete partition, the localized nonlinear dynamics yield a continuous measure of structural strain. The resulting low agreement scores highlight the inter-community edges, offering an endogenous, decentralized indicator of the exact topological fracture lines along which the biological network eventually divided.

\subsection{Evaluating Structural Ambiguity: The College Football Network}
\label{sec:college_football}

To further assess the robustness of the proposed dynamical framework—particularly with severe structural ambiguity—we analyze the standard United States College Football network \cite{girvanCommunityStructureSocial2002b}. This unweighted graph encodes the game schedule of the 2000 Division I season, where vertices represent university teams and edges denote regular-season matchups.

The nominal ground-truth communities in this dataset correspond to administrative athletic conferences (e.g., Big Ten, SEC, Pac-10). Because teams typically play an average of seven intra-conference games compared to only four inter-conference games, these conferences naturally function as dense topological cores, a cohesive structure further reinforced by geographically localized scheduling. The localized nonlinear dynamics drive the vast majority of the network into distinct, stable polarized states, reconstructing these primary conference boundaries via decentralized interactions (Fig. \ref{fig:football_evaluation}).

\begin{figure}[t]
    \centering
    \begin{subfigure}[b]{0.5\textwidth}
        \centering
        \includegraphics{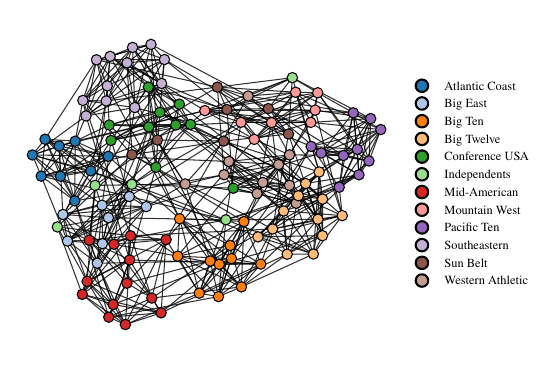}
        \caption{Nominal Ground-Truth Conferences}
        \label{fig:football_gt}
    \end{subfigure}
    \hfil
    \begin{subfigure}[b]{0.35\textwidth}
        \centering
        \includegraphics{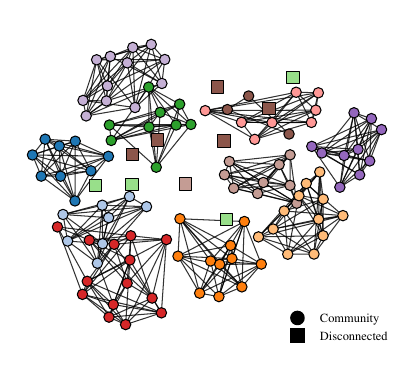}
        \caption{SER Algorithm Filtered State}
        \label{fig:football_ser}
    \end{subfigure}
    \caption{\textbf{Comparison of ground truth and SER output on the College Football Network.} (a) The original network topology, with nodes colored according to their administrative athletic conferences. (b) The emergent partition generated by the SER algorithm ($K=6, h=0.1, \tau=1$) over 100 independent realizations. Surviving edges denote highly reliable structural connections. Nodes marked with an square represent structurally ambiguous teams (primarily from the Sun Belt conference and Independents) that the algorithm decoupled from the cores due to extreme variability between the iterations.}
    \label{fig:football_evaluation}
\end{figure}

However, notable deviations from these administrative labels emerge for the ``Independent'' teams and the ``Sun Belt'' conference. Analytical inspection confirms these discrepancies are not algorithmic failures, but precise reflections of the true topological distribution of these nodes, which lack dense internal connectivity.

Structurally, the ``Independent'' teams operate as isolated boundary agents. They schedule games across a widely distributed array of opponents. Their internal degree with respect to any group is negligible ($d_i^I \approx 0$), while their external degree is disproportionately high. Under the continuous-time update rule \eqref{eq:dynamic_vector}, these nodes lack the internal restoring force to satisfy the strict state confinement conditions from Proposition \ref{prop:cell_deviation_bound}. The local dynamics pull these independent agents into the invariant hypercubes of the denser conferences with which they share the most scheduling ties.

The fragmentation of the Sun Belt conference can be attributed to an analogous topological sparsity. During the sampled season, Sun Belt teams played nearly as many inter-conference games as intra-conference matches \cite{girvanCommunityStructureSocial2002b} and failed to act as a cohesive unit.

When evaluated through the SER framework, the dynamics across multiple independent iterations reflect this ambiguity, resulting in consistently low agreement scores on intra-Sun Belt edges. As illustrated in Figure~\ref{fig:football_ser}, the algorithm safely fractures this sparsely connected group. It decouples the most ambiguous boundary nodes (marked by square) and merges the remaining constituents into the denser conferences with which they most frequently interact.

\section{Conclusion}\label{sec:conclusion}
This paper demonstrates that the macroscopic problem of community detection can be solved entirely through the localized continuous-time interactions of a nonlinear dynamical system. By mapping ``social learning'' to a multi-agent consensus protocol, we established that a saturated nonlinear communication function forces a structurally cohesive network to break symmetry.

Theoretical analysis confirms this localized symmetry breaking is strictly governed by the spectral gap of the network ($\lambda_{N-1}$). When the interaction gain surpasses this spectral threshold, densely connected subgraphs fall into forward-invariant hypercubes, effectively trapping core communities in stable, polarized states.

To apply this opinion dynamics as privacy-preserving distributed algorithms, we introduced three algorithms: Recursive Neighbor Pruning (RNP), a decaying confidence model (RNP-DC), and a Score-based Edge Reliability (SER). SER solves the vulnerability of greedy bisections. It isolates structurally ambiguous frontier nodes and identifies exact community boundaries without requiring global eigenvectors or a priori knowledge of the target number of communities. Empirical validations on the ABCD benchmark and complex real-world topologies (Zachary's Karate Club, the Ngogo Chimpanzee network, and the College Football network) confirm these decentralized protocols achieve clustering accuracy on par with globally optimized centralized heuristics. Our analysis of the LFR benchmark demonstrates the premature spectral collapse of traditional synthetic networks and supports the ABCD model for clustering evaluations.

\bibliography{Community_Detection.bib}

\appendix

\subsection{Proof of Proposition~\ref{prop:iss_cell_dev}}
\begin{proof}
    Consider the sequence $V(t) = \max_{i\in \Vcal'} |x_i(t) - c| = \delta(t)$. Let $m \in \Vcal'$ be the index of the agent realizing this maximum at time $t+1$, such that $V(t+1) = |e_m(t+1)|$.

    Using the discrete-time dynamics \eqref{eq:dynamic_single_agent} and substituting $x_j(t) = c + e_j(t)$, the update for $e_m$ is:
    \begin{align*}
        e_m(t+1) & = e_m(t) + h \left[ \frac{1}{d_m} \sum_{j \in \mathcal{N}_m \cap \Vcal'} a_{mj} \big(s(c + e_j(t)) - c\big) +  \frac{1}{d_m} \sum_{j \in \mathcal{N}_m \setminus \Vcal'} a_{mj}  \big(s(x_j(t)) - c\big) - e_m(t) \right] \\
                 & = (1-h)e_m(t) + h \left[ \frac{1}{d_m} \sum_{j \in \mathcal{N}_m \cap \Vcal'} a_{mj} \big(s(c + e_j(t)) - s(c)\big) + \frac{1}{d_m} \sum_{j \in \mathcal{N}_m \setminus \Vcal'} a_{mj} \big(s(x_j(t)) - c\big) \right],
    \end{align*}
    where we utilized the fixed point property $s(c) = c$. Taking the absolute value, applying the triangle inequality,, noting that $h \in (0,1]$ implies $(1-h) \ge 0$, and $|s(c+e_j(t)) - s(c)| \le L_c|e_j(t)|$ due to the local Lipschitz property on $\mathcal{I}_c$, we obtain:
    \begin{align*}
        V(t+1) \le (1-h)|e_m(t)| + h \left[ L_c \frac{d_m^I}{d_m} V(t) + \frac{d^E_m}{d_m} u_m(t) \right] & \le (1-h)V(t) + h \left[ L_c \mu_m V(t) + (1-\mu_m) \|u(t)\|_{\infty} \right] \\
                                                                                                          & = (1 - h\eta_m) V(t) + h (1-\mu_m) \|u(t)\|_{\infty}.
    \end{align*}
    By the definition of the effective perturbation gain, we have $(1-\mu_m) = \gamma_m \eta_m \le \gamma_{\max} \eta_m$. Substituting this upper bound yields:
    \begin{equation}\label{eq:convex_bound}
        V(t+1) \le (1 - h\eta_m) V(t) + h\eta_m \gamma_{\max} \|u(t)\|_{\infty}.
    \end{equation}
    Let $U(t) = \gamma_{\max} \|u(t)\|_{\infty}$. We prove the discrete-time ISS bound by induction. For $t=0$, the base case holds trivially. Assume $V(t) \le (1-h\eta)^t V(0) + \max_{0 \le \tau < t} U(\tau)$. Substituting this into \eqref{eq:convex_bound}:
    \begin{equation*}
        V(t+1) \le (1-h\eta_m) \Big[ (1-h\eta)^t V(0) + \max_{0 \le \tau < t} U(\tau) \Big] + h\eta_m U(t).
    \end{equation*}
    Because $\eta \le \eta_m$ and $V(0) \ge 0$, the transient term satisfies $(1-h\eta_m)(1-h\eta)^t V(0) \le (1-h\eta)^{t+1} V(0)$. The remaining terms represent a strict convex combination (since $h\eta_m \in (0, 1]$) of the previous maximum perturbation and the current perturbation $U(t)$. A convex combination of two values is strictly bounded by their maximum:
        \begin{equation*}
            (1-h\eta_m) \max_{0 \le \tau < t} U(\tau) + h\eta_m U(t) \le \max_{0 \le \tau \le t} U(\tau).
        \end{equation*}
        Summing these bounds completes the inductive step, verifying the sequence bound for all $t$.

        Finally, to derive the asymptotic limit without trapping initial perturbation peaks, note that the bound holds invariantly for any starting time $T \ge 0$. For any $t > T$:
        \begin{equation*}
            V(t) \le (1-h\eta)^{t-T} V(T) + \gamma_{\max} \max_{T \le \tau < t} \|u(\tau)\|_{\infty}.
        \end{equation*}
        Taking the limit superior of both sides as $t \to \infty$, the transient decay strictly vanishes, leaving:
        \begin{equation*}
            \limsup_{t \to \infty} V(t) \le \gamma_{\max} \sup_{\tau \ge T} \|u(\tau)\|_{\infty}.
        \end{equation*}
        Because this holds for any arbitrarily large initial time $T$, we take the limit as $T \to \infty$. By the definition of the limit superior, $\lim_{T \to \infty} \sup_{\tau \ge T} \|u(\tau)\|_{\infty} = \limsup_{t \to \infty} \|u(t)\|_{\infty}$, which evaluates to the final stated asymptotic bound.
\end{proof}

\subsection{Proof of Proposition~\ref{prop:cell_deviation_bound}}

\begin{proof}
    Suppose at time step $t$, the sub-state vector lies within $\mathcal{H}(c, \varepsilon)$. We must show that for any agent $i \in \Vcal'$, $x_i(t+1) \in [c-\varepsilon, c+\varepsilon]$.

    The discrete update equation can be rewritten as a strict convex combination:
    \begin{equation*}
        x_i(t+1) = (1-h)x_i(t) + h \bar{s}_i(t),
    \end{equation*}
    where $\bar{s}_i(t) = \frac{1}{d_i} \sum_{j} a_{ij} s(x_j(t))$ is the average neighborhood signal. Since $h \in (0,1]$, $x_i(t+1)$ is bounded by the extrema of $x_i(t)$ and $\bar{s}_i(t)$. Because $x_i(t) \le c+\varepsilon$ by assumption, a sufficient condition to ensure $x_i(t+1) \le c+\varepsilon$ is to guarantee that the signal average satisfies $\bar{s}_i(t) \le c+\varepsilon$.

    We split the neighborhood signal into internal and external components:
    \begin{align*}
        d_i \bar{s}_i(t) & = \sum_{j \in \mathcal{N}_i \cap \Vcal'} a_{ij} s(x_j(t)) + \sum_{j \in \mathcal{N}_i \setminus \Vcal'} a_{ij} s(x_j(t)) \leq d_i^I s(c+\varepsilon) + d_i^E s(1),
    \end{align*}
    where we used the fact that $s$ is non-decreasing, $x_j(t) \le c+\varepsilon$ for internal neighbors, and bounded by $1$ for external neighbors. To enforce $\bar{s}_i(t) \le c+\varepsilon$, we require:
    \begin{equation*}
        d_i^I s(c+\varepsilon) + d_i^E s(1) \le d_i (c+\varepsilon) = d_i^I(c+\varepsilon) + d_i^E(c+\varepsilon).
    \end{equation*}
    Rearranging the terms to isolate the external pull yields:
    \begin{equation*}
        d_i^E \big( s(1) - (c+\varepsilon) \big) \le d_i^I \big( (c+\varepsilon) - s(c+\varepsilon) \big).
    \end{equation*}
    By the definitions of $p(c,\varepsilon)$ and $r(c,\varepsilon)$, the outward pull $s(1) - (c+\varepsilon)$ is bounded above by $p(c,\varepsilon)$, and the restoring force $(c+\varepsilon) - s(c+\varepsilon)$ is bounded below by $r(c,\varepsilon)$. Thus, condition \eqref{eq:condition_invariant} strictly implies the inequality, guaranteeing the upper bound.

    The proof for the lower bound $x_i(t+1) \ge c-\varepsilon$ follows symmetrically by establishing $\bar{s}_i(t) \ge c-\varepsilon$ utilizing the minimal external signal $s(-1)$ and the corresponding lower bound of the restoring force. Therefore, $\mathcal{H}(c, \varepsilon)$ is strictly forward invariant.
\end{proof}

\end{document}